\documentclass{article}
\usepackage{amsmath, amsthm, amssymb, graphicx,url,soul,bm}

\newtheorem{theorem}{Theorem}[section]
\newtheorem*{theorem*}{Theorem}
\newtheorem{proposition}[theorem]{Proposition}
\newtheorem*{proposition*}{Proposition}
\newtheorem{lemma}[theorem]{Lemma}
\newtheorem*{lemma*}{Lemma}
\newtheorem{definition}{Definition}[section]
\newtheorem*{definition*}{Definition}
\newtheorem*{claim*}{Claim}
\newtheorem*{question*}{Question}
\newtheorem*{corollary*}{Corollary}
\newtheorem{corollary}[theorem]{Corollary}
\theoremstyle{remark}
\newtheorem{remark}{Remark}
\newtheorem*{remark*}{Remark}

\usepackage{enumitem}
\usepackage[mathscr]{euscript}
\newcommand{\ttc}{{\mathscr C}} 
\newcommand{\ttcopt}{{\mathscr C}_\textnormal{opt}} 
\newcommand{\ttd}{{\mathscr D}}  
\newcommand{\ttu}{{\mathscr U}} 
\newcommand{\tte}{{\mathscr E}}  

\newcommand{\C}{\mathrm C}
\newcommand{\ff}{{\mathrm F}}
\newcommand{\fg}{{\mathrm G}}
\newcommand{\fr}{{\mathrm R}}
\newcommand{\fh}{{\mathrm H}}

\newcommand{\fs}{{\mathrm S}}
\newcommand{\ft}{{\mathrm T}}

\usepackage{relsize} 
\newcommand{\smallerbrackets}[1]{{\mathsmaller{(}#1\mathsmaller{)}}}  

\newcommand{\mcX}{\mathcal{X}}
\newcommand{\mcY}{\mathcal{Y}}
\newcommand{\mcD}{\mathcal{D}}
\newcommand{\mcS}{\mathcal{S}}

\newcommand{\na}{k_\texttt{A}} 
\newcommand{\nb}{k_\texttt{B}}
\newcommand{\pa}{p_\texttt{A}}
\newcommand{\pb}{p_\texttt{B}}
\newcommand{\xa}{x_\texttt{A}}
\newcommand{\xb}{x_\texttt{B}}

\newcommand{\td}{\mathrm{d}}

\newcommand{\tx}{\mathrm{x}}
\newcommand{\ty}{\mathrm{y}}
\newcommand{\tz}{\mathrm{z}}
\newcommand{\tk}{\mathrm{k}}
\newcommand{\tm}{\mathrm{m}}
\newcommand{\tK}{\mathrm{K}}

\newcommand{\Kmax}{{K_\textnormal{max}}}
\newcommand{\kmax}{{k_\textnormal{max}}}
\newcommand{\Lavg}{{L_\textnormal{avg}}}

\newcommand{\maj}{\textnormal{maj}}

\newcommand{\condenser}[2]{$(#1 \mathop{\stackrel{\mathsmaller{\varepsilon}}{\rightarrow}}  #2)$-condenser}
\newcommand{\condensser}[2]{$(#1 \mathop{\stackrel{\mathsmaller{\varepsilon}}{\rightarrow}}  #2)$-conden\textbf{s}er}

\usepackage[margin=20mm]{geometry}
\def\zo{\{0,1\}}
\def\mapping{\rightarrow}
\newcommand{\ie}{{\mbox{i.e.}}}
\newcommand{\nat}{\mathbb{N}}

\usepackage{color}
\definecolor{lightblue}{rgb}{.60,.60,1}
\definecolor{lightred}{rgb}{1, .60, 0.60}
\soulregister\ref7

\newcommand{\wider}{
  \thickmuskip=12mu plus 3mu minus 1mu
  \medmuskip=7mu plus 3mu minus 1mu
  }

\DeclareMathOperator{\poly}{poly}

\DeclareMathOperator*{\ceil}{ceil}

\usepackage{adjustbox} 
\let\oldhash\#%
  \DeclareRobustCommand{\#}{\adjustbox{valign=B,totalheight=.57\baselineskip}{\oldhash\!}}%

\title{Universal almost optimal compression and Slepian-Wolf coding in probabilistic polynomial time}

\author{ 
{Bruno Bauwens\/}
\thanks{National Research University Higher School of Economics, Faculty of Computer Science, Moscow, Russia}
{\quad Marius Zimand\/}
\thanks{  Department of Computer and Information Sciences, Towson University,
Baltimore, MD, USA. \texttt{http://orion.towson.edu/\~{ }mzimand}. Marius Zimand has been supported in part by the National Science Foundation through grant CCF 1811729}}

\begin{document}

\maketitle

\begin{abstract} 
  In a lossless compression system with target lengths,  
  a compressor $\ttc$ maps an integer $m$ and a binary string $x$ to an $m$-bit code $p$, and if $m$ is sufficiently large, 
  a decompressor $\ttd$ reconstructs $x$ from $p$.  We call a pair $(m,x)$ {\em achievable} for $(\ttc,\ttd)$ if this reconstruction is successful.
  We introduce the notion of an optimal compressor $\ttcopt$, by the following universality property: 
  For any compressor-decompressor pair $(\ttc, \ttd)$, there exists a decompressor $\ttd'$ such that if $(m,x)$ is achievable for $(\ttc,\ttd)$, then $(m+\Delta, x)$ 
  is achievable for $(\ttcopt, \ttd')$, where $\Delta$ is some small value  called the overhead. 
  We show that there exists an optimal compressor that has only polylogarithmic overhead and works in probabilistic polynomial time. 
  Differently said, for any pair $(\ttc, \ttd)$, no matter how slow $\ttc$ is, or even if $\ttc$ is non-computable, $\ttc_{opt}$ is a fixed compressor that in polynomial time  produces codes almost as short as those of $\ttc$.  The cost is that the corresponding decompressor is slower.
  
  We also show that each such optimal compressor can be used for distributed compression, in which case  it can achieve optimal compression rates, as given in the Slepian-Wolf theorem, and even for the Kolmogorov complexity variant of this theorem.
  Moreover, the overhead is logarithmic in the number of sources, and unlike previous implementations of Slepian-Wolf coding, meaningful compression can still be achieved if the number of sources is much larger than the length of the compressed strings.
\end{abstract}

\section{Introduction}

Good lossless compression means mapping data to  short compressed codes. Ideally, the codes should  have minimal or   close-to-minimal description length. It is also desirable to  have  efficient compression and decompression procedures.   Unfortunately, it is impossible to obtain close-to-minimal codes from which decompression is even remotely efficient.\footnote{
The standard proof that the Kolmogorov complexity function is not computable provides such counter examples. Indeed, for the sake of contradiction, assume that for some computable time bound $t\colon \mathbb{N} \rightarrow \mathbb{N}$, the following holds: For all $n$ and $n$-bit $x$ with complexity smaller than $2\log n$, there exists a program of length at most $n/10$ that generates $x$ in time $t(n)$. Then the lexicographically first $n$-bit string $x$ that does not have such a program, has complexity at most $\log n + O(1)$, which is a contradiction for large~$n$.} 
In contrast, we show that, remarkably, there is no fundamental conflict between having close-to-minimal codes and efficient compression. More precisely, if a good approximation of the length of a minimal description of the data is given, then a probabilistic compressor can generate an almost minimal compressed code in polynomial time. The same optimal compressor can be used for any description system, as explained in the abstract. Moreover, it can be used for distributed compression as well.

Our study of compression uses a more general approach than what is standard in information theory. In information theory, the  underlying assumption is that data is generated by a stochastic process, whose structural properties, such as various types of independence, ergodicity, various statistics of the distribution, etc., are at least partially used in constructions or in the analysis. Instead, we do not assume any generative model, and require that for {\em each} string, the compressed length is as small as in any other compressor-decompressor pair up to an additive polylogarithmic term. Our main results subsume important results in information theory: the Shannon source coding theorem and the Slepian-Wolf coding theorem. Moreover, they provide corresponding results in algorithmic information theory, where the optimal compression length of a string is measured by its Kolmogorov complexity relative to an optimal Turing machine.

The contributions of this paper concern both single-source compression and multi-source distributed compression.

\subsection*{Single source compression}

By direct diagonalization, one may observe that no (probabilistically) computable compressor is optimal: for each such compressor $\ttc$ mapping binary strings to binary strings,
there exists another deterministic compressor $\ttc'$ and a sequence $x_1, x_2, \dots$ of strings with length $|x_n| = n$ such that $|\ttc(x)| \ge n$ (using the maximal length over all randomness if $\ttc$ is probabilistic) and $|\ttc'(x)| \le O(\log n)$.  To overcome this barrier, we consider compression functions that have a ``target compression size'' as an additional input. Our compressors are also probabilistic. 

Consequently, a compressor $\ttc$ is a probabilistic  algorithm that has three inputs: the string~$x$ being compressed, the target compression length~$m$,  and the error probability~$\varepsilon$.  The output is a random variable, denoted $\ttc_{\varepsilon,m}(x)$, whose realization is a string of length $m$.

A  {\em decompressor} $\ttd$ is a deterministic partial function mapping strings (viewed as compressed codes) to strings. The {\em Kolmogorov complexity} of $x$ relative to the decompressor $\ttd$, is given by
\begin{equation*}
\C_{\ttd}(x) \;=\; \min \{|p|: \ttd(p) = x\}. 
\end{equation*}
Note that unlike standard Kolmogorov complexity theory, we do not restrict $\ttd$ to be partial computable. 
Our definition of optimal compressors requires that 
($i$) the length of the compressed code is equal to the target length, and  
($ii$) simultaneously for {\em all} decompressors~$\ttd$, correct decompression is achieved as soon as $m$ is 
slightly above $\C_\ttd(x)$.

\begin{definition}\label{def:optimal_compressor}
  A {\em compressor} $\ttc$ is a probabilistic function that maps 
  $\varepsilon>0$, $m$ and $x$, with probability~$1$,
  to a string $\ttc_{\varepsilon,m}(x)$ of length exactly~$m$. 
  Let $\Delta$ be a function of $\varepsilon, m$ and $x$, called {\em overhead}.
  A compressor $\ttc$ is {\em $\Delta$-optimal} if for every decompressor $\ttd$ there exists 
  a decompressor $\ttd'$ such that for all $\varepsilon > 0$, $x$ and $m \ge \C_\ttd(x) + \Delta$: 
 \[
 \Pr \big[ \ttd'\big( \ttc_{\varepsilon,m}(x) \big) = x    \big] \;\ge\; 1-\varepsilon.
 \]
\end{definition}

\noindent
In our constructions, the mapping from $\ttd$ to $\ttd'$ is effective in a certain sense,
and this implies that if $\ttd$ is partial computable, respectively, computable, and computable in polynomial space, 
then so is $\ttd'$, see Remark~\ref{remark:effective_decompression}.

\medskip
\noindent
Our main result for single-source compression states that for some polylogarithmic overhead $\Delta$, an optimal compression function exists, and can even be computed in polynomial time.

\begin{theorem}
  \label{th:main_single_source}
  There exists a $O(\log m \cdot \log \tfrac{|x|}{\varepsilon})$-optimal compressor that is computable in time polynomial in~$|x|$.
\end{theorem}

\noindent
In other words, any slow compressor in a compressor-decompressor pair can be replaced by a fixed fast compressor and slower decompressor, 
at the expense of a small increase of the compression length.

\smallskip
\begin{remark*}
  It is standard in compression to assume that the parties do not share a source of randomness. If they do, the result  of  Theorem~\ref{th:main_single_source} can be obtained via simple hashing, as we explain in section~\ref{s:techniques}. 
\end{remark*}

\smallskip
\noindent
In Kolmogorov complexity, one typically fixes an optimal\footnote{
  $\ttu$ is a Turing machine, and its optimality means that for every other Turing machine $\ttd$, there exists a constant $c$ such for every string $x$, $\C_{\ttu}(x) \le \C_{\ttd}(x) + c$. Here, we assume that $\ttu$ is {\em universal} in the sense that for each other Turing machine~$\ttd$, there exists a string $w$ such that $\ttu(wp) = \ttd(p)$ for all strings $p$ for which $\ttd(p)$ is defined.} 
  partial computable decompressor $\ttu$ and writes $\C(x)$ in lieu of $\C_\ttu(x)$. 
  $\C(x)$ is called the {\em Kolmogorov complexity} or the {\em minimal description length} of~$x$. 
  In a similar way, for two strings $x$ and $y$, we define $\C(x \mid y)$, the complexity of $x$ conditional on~$y$, see equation~\eqref{e:kolmcond}. 
Theorem~\ref{th:main_single_source} implies polynomial time compression of a string $x$ down to almost its minimal description length, more precisely down to length $m = \C(x) + O(\log^2 (|x|/\varepsilon))$, provided the compressor works with target length~$m$. 

\begin{corollary}\label{cor:main_single_source}
  There exists a polynomial-time computable compressor $\ttc$ and a constant $c$, 
  such that for all $\varepsilon > 0$, $m$ and $x$ with $C(x) \le m - c \cdot \log m \cdot \log \tfrac{|x|}{\varepsilon}$: 
  \[
  \Pr [\ttu(\ttc_{\varepsilon,m}(x)) = x ] \;\ge\; 1-\varepsilon.
  \]
\end{corollary}

\noindent
The compressor in Theorem~\ref{th:main_single_source}  also provides a solution to the so-called \emph{document exchange problem}, which can be formulated as follows. Suppose $y$ is the obsolete version of an updated file $x$. The receiver holds $y$ and the sender holds $x$.  The sender transmits  $\ttc_{\varepsilon,m}(x)$ to the receiver, and  if $\C(x\mid y) \le m - \Delta$, 
then  the receiver can reconstruct~$x$.
This follows from the optimality condition in  Definition~\ref{def:optimal_compressor} applied for $\ttd(\cdot) = \ttu(\cdot,y)$, 
where $\ttu$ is the optimal machine in the definition of~$\C(\cdot \mid \cdot)$. 
It is remarkable that the sender computes  $\ttc_{\varepsilon,m}(x)$ without knowing~$y$, but only the target length $m=\C(x\mid y) + \Delta$.  

\bigskip
\noindent
The overhead $\Delta$ can be improved from polylogarithmic to logarithmic at the cost of a slower compressor running in exponential space, and, hence,  double exponential time.  

\begin{proposition}\label{prop:single_source_nonpoly}
  There exists a $O(\log \tfrac{|x|}{\varepsilon})$-optimal compressor that is computable in exponential space.
\end{proposition}

\noindent
The compressor in Proposition~\ref{prop:single_source_nonpoly} is also probabilistic and uses a logarithmic number of random bits. 
We prove the following lower bounds, which show that the overhead and the randomness  in Proposition~\ref{prop:single_source_nonpoly} are essentially minimal: Every compressor with a computably bounded running time has overhead $\Delta \ge \Omega(\log\tfrac{|x|}{\varepsilon})$, and if $\Delta \le O(|x|^{0.99})$, at least a logarithmic amount of randomness is needed. Let $\lceil a \rceil = \ceil(a)$.

\begin{proposition}\label{prop:lowerbounds} 
  Let $\Delta$ and $r$ be functions of~$\varepsilon$ and~$n=|x|$.
  If there exists a $\Delta$-optimal compressor that can be evaluated with at most $r$ random bits in a computably bounded running time, then
  \begin{itemize}[leftmargin=*]
    \item 
      $
      \Delta \ge \log \tfrac{n}{\varepsilon} - \log \log \tfrac{n}{\varepsilon} - 8 
      $
      for all $n$ and $\varepsilon$ with $2^{-n/4} \le \varepsilon \le 1/4$,

    \item 
      $
      \left\lceil \varepsilon 2^{r+1} \right\rceil \;\ge\; \frac{n-r - \log (2/\varepsilon)}{\Delta + 4}
      $
      for all $n$ and $\varepsilon \le 1/2$.
  \end{itemize}
\end{proposition}

\noindent
The second item for $r=0$ and $\varepsilon = 1/2$ implies that a deterministic computable compressor can only be $\Delta$-optimal for $\Delta \ge n - 6$.  
Essentially the same bounds also apply for a weaker version of $\Delta$-optimality in which we consider a single optimal Turing machine as decompressor, and allow probabilistic decompression, see Lemmas~\ref{lem:equivalence_optimality_probabilisticDecompressor} and~\ref{lem:equivalence_optimality_singleTM}.

\subsection*{Distributed compression}

In its most basic form,  distributed compression is the task of compressing correlated pieces of data by several parties, each one possessing one piece  and acting separately.  This task is also known as \emph{Slepian-Wolf coding.}

For illustration, let us consider the following simple example,  with two senders, Alice and Bob, and a receiver, Zack.  
Alice knows a line $L$ of the form $y = ax + b$ in the affine plane over a finite field with $2^n$ elements, 
and Bob knows a point $P$ with coordinates  $(u, v)$ on this line.
Each of $L$ and $P$ has $2n$ bits of information, but, because of the geometrical relation,
together they only have $3n$ bits of information. 
Alice and Bob want to email $L$ and $P$ to Zack, without wasting bandwidth.
If they have to compress the message in isolation, how many bits 
should they  send to Zack?

Obviously, Alice can send $L$, so $2n$ bits,  and then Bob can send
$n$ bits, or symmetrically Bob can send $P$, so $2n$ bits, and then Alice can send $n$ bits.
What about other compression lengths?  Some necessary requirements for the
compression lengths are easy to see.  Let $\na$ be the number of bits to which
Alice compresses her point $L$, and let $\nb$ have the analogous meaning for
Bob.  It is necessary that 
$\na + \nb  \geq 3n$,  $ \na \geq n$  and $\nb \geq n$.
The first inequality holds because Zack needs to acquire $3n$ bits, and the other two hold
 because if Zack gets somehow either $L$ or $P$, he still needs $n$ bits of information  about the remaining element.

Without any assumption regarding how $L$ and $P$ were generated, 
the version of Slepian-Wolf coding in this paper implies that  any numbers $\na$ and $\nb$
satisfying the above necessary conditions 
(such as, for instance,  $\na = \nb = 3n/2$), 
are also sufficient up to a small polylogarithmic overhead. 
More precisely, there exists  a probabilistic polynomial-time  compression algorithm
such that if $\na$ and $\nb$ satisfy these conditions, then Alice can apply the algorithm 
to compress the line $L$ to a binary string $\pa$ of length $\na + O(\log^2 n)$, 
Bob can compress point $P$ to a binary string $\pb$ of length $\nb + O(\log^2 n)$,  
and Zack can with high probability reconstruct $L$ and $P$  from $\pa$ and $\pb$.  

In the general case (presented here for simplicity for two senders), Alice has  a string $\xa$, Bob has $\xb$, which they want to send to Zack. Suppose Zack is using the decompressor $\ttd$. Let $\na, \nb$ represent the compression lengths as above, and let $\C_\ttd(\xa, \xb), \C_\ttd(\xa \mid \xb), \C_\ttd(\xb \mid \xa)$ denote the Kolmogorov complexities relative to $\ttd$ of $(\xa, \xb)$, and respectively, of $\xa$ conditioned by $\xb$ and of $\xb$ conditioned by $\xa$.  Similarly to the simple example with the line and the point, for Zack to be able to reconstruct $(\xa, \xb)$, it is necessary (up to a small additive term) that $\na$ and $\nb$ satisfy the following inequalities, called the Slepian-Wolf constraints:
\begin{equation}
\label{e:kswconditions}
\begin{array}{ll}
\na &\geq\; \C_\ttd(\xa \mid \xb)\\
\nb &\geq\; \C_\ttd(\xb \mid \xa)   \\
\na  + \nb &\geq\; \C_\ttd(\xa, \xb).
\end{array}
\end{equation}
Note that these requirements on the compression lengths are necessary even
when Alice and Bob are allowed to collaborate.  
We show that modulo a polylogarithmic overhead, the conditions are also sufficient.
Moreover, such compression is achieved by the optimal compressors from Definition~\ref{def:optimal_compressor}.

Formally, let $\ttd$ be a decompressor mapping pairs of strings to strings.
The {\em conditional} Kolmogorov complexity of $w$ given $z$ is 
\begin{equation} \label{e:kolmcond}
\C_\ttd(w|z) \;=\; \min \left\{|p|: \ttd(p,z) = w  \right\}.
\end{equation}
Let $[\ell] = \{1,2,\dots,\ell\}$. The parameter $\ell$ represents the number of senders.  Given $J \subseteq [\ell]$, and an $\ell$-tuple $\tx = (x_1, \dots, x_\ell)$, let $\tx_J$ be the set $\{(j,x_j): j \in J\}$. 
The Kolmogorov complexity of $\tx_J$ is defined by encoding this finite set as a string in some canonical way. 

\begin{definition}\label{def:SW_rate}
  An $\ell$-tuple $\tk = (k_1, \ldots, k_\ell)$ of integers satisfies the {\em $\ttd$-Slepian-Wolf constraints} for $\tx$ if 
 \[
 \C_\ttd\left(\tx_J \mid \tx_{[\ell]\setminus J} \right) \;<\;  \sum_{j\in J} k_j \quad \text{ for all non-empty } J.
 \]
\end{definition}

\begin{theorem}\label{th:main_distributed_compression}
  Assume $\ttc$ is $\Delta$-optimal.
  For every decompressor $\ttd$ there exists a decompressor $\ttd'$ such that for all $\varepsilon> 0$, $\ell$, 
  $\ell$-tuples $\tx$, and $\ell$-tuples $\tm$: 
 \[
 \Pr \big[ \ttd'\big(\varepsilon, \ttc_{\varepsilon,m_1}(x_1), \dots, \ttc_{\varepsilon,m_\ell}(x_\ell) \big) = \tx    \big] \;\ge\; 1- 8\ell\varepsilon,
 \]
 provided $(m_1 {-} \Delta_1 {-}\log \tfrac{\ell}{\varepsilon}, \;\dots,\ m_\ell {-} \Delta_\ell {-} \log \tfrac{\ell}{\varepsilon})$ 
 satisfies the $\ttd$-Slepian-Wolf constraints for $\tx$, where
 $\Delta_j$ is the value of $\Delta$ corresponding to the inputs $\varepsilon$, $m_j$ and $x_j$.
\end{theorem}

\noindent
The compressors in this definition work in isolation in a very strong sense: besides their string $x_j$, 
they only use the targets $\varepsilon$ and~$m_j$. They do not use the decompressor~$\ttd$ or even~$\ell$. 

If we apply this theorem to the compressor in the proof of Theorem~\ref{th:main_single_source}, 
then a similar remark for the effectivity of the transformation from $\ttd$ to $\ttd'$ 
holds as for Definition~\ref{def:optimal_compressor}:
if $\ttd$ is partial computable, respectively, computable, and computable in exponential time,
then so is $\ttd'$, see Remark~\ref{remark:effective_decompression_distributed}.
Moreover, for growing $\ell$, the overhead $\Delta_j + \log \tfrac{\ell}{\varepsilon}$ is logarithmic in $\ell$ and the error is linear in~$\ell$.
Hence, the result can be meaningful even when the number of senders is exponential in the length of the strings.

\subsection*{Related works and comparison with our results}\label{s:relatedworks}

The study of the fundamental limits of data compression has been undertaken in both Information Theory (IT) and Algorithmic Information Theory (AIT). We recall that the major conceptual difference between the two approaches is that in IT  data is viewed as being the outcome of  a generative process, whereas in AIT the data consists of individual strings without any assumption regarding their provenance. One can say that IT focuses on \emph{random processes}, that is processes whose outcomes have a degree of uncertainty, and AIT focuses on \emph{random strings}, that is strings that do not admit a concise description due to the lack of discernible patterns. Our approach is in the AIT spirit, but is more general, because it is not restricted to  computable decompressors. 

\smallskip
We start by reviewing basic results in IT. For concreteness, let us first suppose that an $n$-bit string is produced by a single activation  of some generative process (this is called the ``one-shot" scenario). The process is represented by a random variable $X^{(n)}$ which takes values in $\zo^n$. A zero-error compression procedure  is an injective mapping $\ttc: \zo^n \mapping \zo^{*}$ and if $\ttc(u)=v$ we think that $v$ is the compressed encoding of $u$. The goal is to minimize $\Lavg(\ttc)$, which is the expected length of $\ttc(X^{(n)})$. If the range of $\ttc$ is a prefix-free set (or, more generally, a so-called instantaneous code), which is desirable if we want to extend compression to a ``multiple-shot" setting, then the Kraft inequality implies that the average length has to be at least the Shannon entropy of $X^{(n)}$, \ie,  $\Lavg(\ttc) \ge H(X^{(n)})$. 

Huffman coding~\cite{huf:j:compression} is an algorithm that computes a prefix-free coding $\ttc_{\textnormal{Huffman}}$ with minimal average length satisfying  $H(X^{(n)}) \le \Lavg(\ttc_{\textnormal{Huffman}}) \le H(X^{(n)} )+ 1$. The drawback is that Huffman compression takes time exponential in $n$ and requires the knowledge of the distribution of $X^{(n)}$.  

One way to overcome these issues it to depart from the ``one-shot" scenario and to consider that the data is produced by some stochastic process $X_1, X_2,  \ldots $ which often satisfies simplifying assumptions, and also to allow some error.  The standard stochastic process of this type is the \emph{memoryless source}, which means that $X^{(n)} = (X_1, X_2, \ldots, X_n)$ are $n$ $\mbox{i.i.d.}$ draws from some distribution $p_x$ on $\zo$. A fixed-length code with  compression rate $R$ is given by a family of functions $(\ttc_n, \ttd_n)_{n \in \nat}$ of type $\ttc_n : \zo^n \mapping \zo^{Rn}$, $\ttd_n : \zo^{Rn} \mapping \zo^{n}$. The error of the code is defined to be $\varepsilon_n = {\rm Prob} (\ttd_n(\ttc_n(X^{(n)})) \not = X^{(n)})$. 

The Shannon Source Coding Theorem~\cite{sha:j:communication} shows that if $R > H(p_x)$ then there exists a code
$(\ttc_n, \ttd_n)_{n \in \nat}$ with compression rate $R$ and $\lim_{n \rightarrow \infty} \varepsilon_n = 0$, and if $R < H(p_x)$, then every code $(\ttc_n, \ttd_n)_{n \in \nat}$ with compression rate $R$  has $\inf \varepsilon_n > 0$. The compression in the positive part of the Source Coding theorem takes time exponential in $n$ and, since the error is not zero,  the decompression fails for some strings. 

\smallskip
Similarly to the Source Coding theorem for single-source compression, the Slepian-Wolf theorem~\cite{sle-wol:j:distribcompression} characterizes the possible compression lengths for \emph{distributed  compression} in the case of memoryless sources.
For notational convenience,  our short description of the Slepian-Wolf theorem is restricted  to the case of  two senders (Alice and Bob) and of data represented in the binary alphabet.  Recall that in distributed compression, Alice observes $X^{(n)}$, and Bob observes $Y^{(n)}$, where $(X^{(n)}, Y^{(n)})$ have a joint distribution over $\zo^n \times \zo^n$. 

In the case of a \emph{2-Discrete Memoryless Source} (2-DMS), $(X^{(n)}, Y^{(n)})$ are obtained by  $n$ $\mbox{i.i.d.}$ draws from some distribution $p_{x,y}$ on $\zo \times \zo$. In other words, a 2-DMS is given by a sequence $(X_1, Y_1), (X_2, Y_2), \ldots$, where $(X_i, Y_i)$ is the $i$-th independent draw from $p_{x,y}$. Alice observes $X^{(n)}= (X_1, \ldots, X_n)$ and Bob observes $Y^{(n)} = (Y_1, \ldots, Y_n)$.  A distributed compression procedure  with compression rates $(R_1, R_2)$ consists of a family of functions $(\ttc_{1,n}, \ttc_{2,n}, \ttd_n)_{n \in \nat}$ of the type
$\ttc_{1,n} : \zo^n \mapping \zo^{R_1 n}$, $\ttc_{2,n} : \zo^n \mapping \zo^{R_2 n}$, $D_{n} : \zo^{R_1 n} \times \zo^{R_2 n} \mapping \zo^{n} \times  \zo^{n}$. The error is defined as $\varepsilon_n = {\rm Prob} (\ttd_n (\ttc_{1,n}(X^{(n)}), \ttc_{2,n}(Y^{(n)})) \not= (X^{(n)}, Y^{(n)}))$. 

The compression rates $(R_1, R_2)$ are \emph{achievable} if $\lim_{n \mapping \infty} \varepsilon_n = 0$, and the \emph{achievable rate region} is the closure of the set of achievable rates. The Slepian-Wolf theorem states that for a 2-DMS as above the achievable rate region is given by all $(R_1, R_2)$ satisfying $R_1 \geq H(X_1 \mid Y_1)$, $R_2 \geq H(Y_1 \mid X_1)$, and $R_1 + R_2 \geq H(X_1, Y_1)$ (these inequalities form the Slepian-Wolf constraints).

\smallskip
The Source Coding Theorem has been extended to sources that are stationary and ergodic (using their  asymptotic equipartition property~\cite[Theorem 16.8.1]{cov-tho:b:inftheory}), and even to  arbitrary sources~\cite{han-ver:j:infspectrum}, using the \emph{information-spectrum concepts} initiated by  Han and Verd\'{u} (see the monograph~\cite{han:b:infspectrum}). 

Similar extensions exist for the Slepian-Wolf theorem: it has been generalized to sources that are stationary and ergodic~\cite{cov:j:slepwolfergodic}, and, using the  information-spectrum framework, a Slepian-Wolf coding theorem has been obtained for general sources~\cite{miy-kan:j:slepwolfinfspectrum}. However, these  latter results involving non-memoryless sources have a strong asymptotic nature and require that the distribution of $X^{(n)}$ (and, respectively,  of $(X^{(n)}, Y^{(n)})$ for Slepian-Wolf compression) is known. 

\medskip
\noindent
In summary, the IT results mentioned above show that for data produced by a stochastic process: 
\begin{itemize}[leftmargin=*]
  \item
 the minimal compression length for single-source compression  is arbitrarily close to Shannon entropy in case of a memoryless source, or to some asymptotical version of entropy for more general sources;  in the case of distributed compression, the optimal compression lengths are dictated by the Slepian-Wolf constraints, and 

 \item 
 optimal compression relies  either on  the very simple memoryless property, or on complete knowledge of the distribution.  
\end{itemize}

\noindent
Theorem~\ref{th:main_single_source} and  Theorem~\ref{th:main_distributed_compression}  can be viewed as a strengthening  of the Source Coding and Slepian-Wolf theorems. The main merit of our results compared to the classical theorems  is that the compressor works in polynomial time, and is universal in the strong sense of Definition~\ref{def:optimal_compressor}, which, in particular,  means that it does not assume any generative model.  

Another significant difference is that whereas in  the IT results mentioned above (with the exception of Huffman coding)   compression fails for some realizations of the sources, the compression in our results  works for all input tuples satisfying the promise condition. 

Our theorems imply their IT counterparts for memoryless sources and for stationary and ergodic sources using the fact 
that  a sequence of i.i.d. random variables (or, more generally, the outcome of a stationary ergodic source)
gives with a high probability a string of letters whose  Kolmogorov complexity is close to Shannon's entropy of the random source (this was stated in \cite[Proposition~5.1]{zvo-lev:j:kol} and a proof appears in~\cite{brudno1982entropy}; see also~\cite{ming2014kolmogorov, horibe2003note}).

\smallskip
Compression procedures for individual inputs (\ie, without using any knowledge regarding the generative process) have been previously designed using the celebrated Lempel-Ziv methods~\cite{lem-ziv:j:compress,ziv:j:compressindivid}.  This approach has lead to efficient compression algorithms that are used in practice.  Such methods have been used for distributed compression as well~\cite{ziv:j:compresshelper,due-wol:j:slep-wolf-individ,kas:j:slep-wolf-individ}.
For such procedures two kinds of optimality have been established, both valid for infinite sequences and thus having an asymptotic nature. 

First,  the procedures achieve a compression length that is asymptotically equal to the so-called finite-state complexity,  which is the minimum length that can be  achieved by finite-state encoding/decoding procedures~\cite{ziv:j:compressindivid}.  

Secondly, the compression rates are asymptotically optimal in case the infinite sequences are generated by sources that are stationary and ergodic~\cite{wyn-ziv:j:compress}.  

In contrast, the compression in Theorem~\ref{th:main_single_source} applies to finite strings and achieves a compression length close to minimal description length.  Unfortunately, it  cannot be practical, because, as we have explained, when compression is done at this level of optimality, decompression requires time longer than any computable function.  Our theoretical results show that compression is not a fundamental obstacle in the design of an efficient compressor-decompressor pair  that is optimal in some rigorous sense, and we hope that this principle will inspire and guide future research lines with applicative objectives. 

\bigskip
\noindent
We now move to  previous results obtained in the AIT framework. 
How well can we compress a string $x$? By a simple counting argument, there is no compressor (even probabilistic, and even incomputable) that compresses all strings $x$  down to length $\C(x) - c$, for some constant $c$.  As already mentioned no computable procedure can compress $x$ down to a string of length $\C(x)$. 

On the other hand, an easy argument shows that, if $\C(x)$ is also given (such an additional information regarding the input is known as \emph{help}), then  a description of $x$ of length $\C(x)$  can be obtained by exhaustive search, 
which, unfortunately, is an exceedingly slow procedure whose running time grows faster than any computable function. 
Corollary~\ref{cor:main_single_source} shows that in fact a description of $x$ of  almost minimal length can be found in probabilistic polynomial time. 

An interesting problem is the compression of $x$ conditioned by some unavailable information $y$, also known as,  the \emph{asymmetric} version of the Slepian-Wolf coding or 
\emph{source coding with side information at the receiver} in the information-theory literature, and 
as \emph{information reconciliation} or the \emph{document exchange problem} in the theoretical computer science literature.  
%
This time, a  compressor which knows $\C(x \mid y)$, but not $y$, cannot perform the exhaustive search. 

Muchnik's theorem~\cite{muc:j:condcomp} gives a surprising solution: There exist algorithms $\ttc$ and $\ttd$ such that for all $n$ and 
for all $n$-bit strings $x$ and $y$,  $\ttc$ on input $x$, $\C(x \mid y)$ and $O(\log n)$ help bits outputs a string $p$ of  length $\C(x \mid y)$,   and $\ttd$ on input $p$, $y$, and $O(\log n)$ help bits reconstructs $x$. 
Thus, given $\C(x \mid y)$,  Alice can compute from her string $x$  and only $O(\log n)$ additional help bits 
a string $p$ of minimum description length such that Zack using $p$, $y$ and $O(\log n)$ help bits can reconstruct~$x$.

Muchnik's theorem has been strengthened in several ways. Musatov, Romashchenko and Shen~\cite{mus-rom-she:j:muchnik} have obtained a version of Muchnik's theorem for space bounded Kolmogorov complexity, in which both compression and decompression are space-efficient.  Romashchenko~\cite{rom:j:slepwolf} has extended Muchnik's theorem to the general (\ie, non-asymmetric) Slepian-Wolf coding. His result is valid for any constant number of senders, but, for simplicity, we present it for the case of two senders: For any two $n$-bit strings $x$ and $y$ and any two numbers $n_1$ and $n_2$ satisfying the Slepian-Wolf constraints~\eqref{e:kswconditions}, there exist two strings $p_1$ and $p_2$ such that $|p_1| = n_1 + O(\log n), |p_2| = n_2 + O(\log n), \C(p_1 \mid x) = O(\log n), \C(p_2 \mid y) = O(\log n)$ and $\C(x,y \mid p_1, p_2) = O(\log n)$. In words, for any $n_1$ and $n_2$ satisfying the Slepian-Wolf constraints, Alice can compress $x$ to a string $p_1$ of length just slightly larger than $n_1$, and Bob can compress $y$ to a string $p_2$ of length just slightly larger than $n_2$ such that Zack can reconstruct $(x,y)$ from $(p_1, p_2)$, provided all the parties  use $O(\log n)$ help bits.
\if01
 These results raise the following questions: (a) can the help bits be eliminated?,
\footnote{In Muchnik's theorem, Alice computes a program $p$ of minimum description length such that $U(p,y)=x$ from $x$, $\C(x \mid y)$ and $O(\log n)$ help bits, where $U$ is the universal Turing machine underlying Kolmogorov complexity. One can hope to eliminate the $O(\log n)$ help bits (as we ask in question (a)), but not the $\C(x \mid y)$ component. This is not possible even when $y$ is the empty string. Indeed, it is known that for some strings $x$, the computation of $\C(x)$ from $x$, and therefore also the computation of a short program $p$ for $x$, requires that some information of size $\log|x|-O(1)$ bits is available~\cite{bau-she:j:compcomp,gac:j:symmetry}.}
and (b) is it possible to  implement the protocol efficiently, 
i.e., in polynomial time?  
\fi

Bauwens et al.~\cite{bmvz:c:shortlist}, Teutsch~\cite{teu:j:shortlists} and Zimand~\cite{zim:c:shortlistshortproof}  have obtained versions of Muchnik's theorem with polynomial-time compression, but in which the help bits are still present. In fact, their results are stronger in that the compression  procedure on input $x$ outputs a polynomial-size list of strings guaranteed to contain a short program for $x$ given $y$. This is called list approximation. Note that using $O(\log n)$ help bits, the compression procedure can pick the right element from the list, re-obtaining Muchnik's theorem. The gain is that this compression procedure halts even with incorrect help bits, albeit in this case  the  result may not be the desired $x$. 


Bauwens and Zimand~\cite{bau-zim:c:linlist} show the existence of  polynomial-time list approximation algorithms for programs of minimal description lengths  that do not require help bits, but which, instead, are probabilistic. Zimand~\cite{zim:stoc:kolmslepwolf}  obtains Slepian-Wolf coding in the Kolmogorov complexity setting via probabilistic polynomial-time encoding without help bits. 
  
\smallskip
We were inspired by~\cite{bau-zim:c:linlist}  and~\cite{zim:stoc:kolmslepwolf},  which represent the starting point for our study. This work  contains conceptual and technical ideas that depart significantly from these two papers.  The main conceptual novelty is that the optimal compressors in this paper are universal, i.e., they work for any decompressor, computable or not, in the sense explained in Definition~\ref{def:optimal_compressor}.  

On the technical side, we use a new fingerprinting technique, based on condensers and conductors, which leads to smaller overhead  $(\log^2 n$, compared to $\log^3 n$ in ~\cite{zim:stoc:kolmslepwolf}).  

Also, the Slepian-Wolf coding in Theorem~\ref{th:main_distributed_compression} provides a double exponential improvement of  the dependence of the overhead on the number $\ell$ of senders compared to~\cite{zim:stoc:kolmslepwolf}. The latter paper implicitly assumes that $\ell$ is constant; otherwise,  the overhead obtained there makes the result essentially  meaningless. The same is true for the classical Slepian-Wolf theorem: no non-trivial compression can be achieved if the number of senders is significantly larger than the length of the compressed strings.  One of our new ideas is the use of {\em random tree partitioning}, and this technique can also be applied to the standard proof of the classical result to obtain an exponential improvement on the number~$\ell$ of sources.

\section{Optimal single source compression: Proof of Theorem~\ref{th:main_single_source}}
\label{s:techniques}

The main novel features of the proof  are  an efficient fingerprinting construction based on condensers and conductors and a method to handle fingerprints that produce many collisions. We present below the basic ideas.  

\subsection*{Invertible functions}

In Definition~\ref{def:optimal_compressor}, the optimal  compressor and the corresponding decompressor $\ttd'$ only need to work correctly for strings $x$ in 
\begin{equation}\label{e:sets}\tag{*}
  \{x : \C_D(x)< k\},
\end{equation}
where $k = m - \Delta + 1$. 
For different $\ttd$, these sets are different, however all of them have size less than~$2^k$. 
This set contains the initial suspects, from which the decompressor has to find the compressed string~$x$.
The optimal compressor provides a {\em fingerprint}, i.e., a small amount of information 
about the input string $x$ that allows its  identification among these suspects. Because of the universality property, the compressor needs to generate fingerprints without knowing the list of suspects. For any set $\mcX$, let $\mcX^{\le K}$ denote the set of lists of elements in $\mcX$ with size at most~$K$. 
The above discussion leads to the following definition.

\begin{definition}\label{def:invertible}
  A probabilistic function $\ff \colon \mcX \rightarrow \mcY$ is {\em $(K,\varepsilon)$-invertible} if there exists a deterministic partial function 
  \mbox{$g \colon \mcX^{\le K} \times \mcY \rightarrow \mcX$} such that for all $S \in \mcX^{\le K}$ and all~$x \in S$:
 \[
   \Pr \left[g_S(\ff(x)) = x\right] \;\ge\; 1-\varepsilon\,,
 \]
 where $g_S(y) = g(S,y)$.
 $\ff$ is  {\em online} $(K,\varepsilon)$-invertible if there exists such a function $g$ that is {\em monotone} in $S$: 
 if list $S'$ extends $S$,  then the function $y \mapsto g_{S'}(y)$ is an extension of~$y \mapsto g_S(y)$.
\end{definition}

\noindent
The interpretation is that $g$, on input a list $S$ of suspects and a random  fingerprint  $\ff(x)$ of $x$,   identifies $x$ with high probability  among the suspects. The main technical result, from which Theorem~\ref{th:main_single_source} follows, is a probabilistic polynomial time algorithm that computes an invertible function.

\begin{theorem}\label{th:main_invertible}
  There exists  a probabilistic algorithm $\ff$ that on input $\varepsilon>0$, $k$ and string $x$, outputs  
  in time polynomial in~$|x|$ a string $\ff_{\varepsilon,k}(x)$ of length $k+O(\log k \cdot \log (|x|/\varepsilon))$, 
  such that for all $\varepsilon > 0$ and $k$, the function $x \mapsto \ff_{\varepsilon,k}(x)$ is 
  online $(2^k,\varepsilon)$-invertible. 

  Moreover, there exists a family of monotone inverses $g_{\varepsilon,k}$ of $\ff_{\varepsilon,k}$, 
  for which the mapping $(\varepsilon, k, S, y) \mapsto g_{\varepsilon,k,S}(y)$ can be evaluated in space polynomial in $\max_{z \in S} |z|$. 
\end{theorem}

\begin{proof}[Proof of Theorem~\ref{th:main_single_source} (assuming Theorem~\ref{th:main_invertible}).] 
  Without loss of generality, we assume that $\varepsilon$ is a negative power of $2$, because rounding down $\varepsilon$ increases the overhead by less than a constant factor. Hence, we can represent $\varepsilon$ in binary using at most $O(\log \tfrac{1}{\varepsilon})$~bits.
  The optimal compressor outputs the string $\ttc_{\varepsilon,m}(x)$ which is an $m$-bit representation of the triple 
  $(\varepsilon, k(\varepsilon,m,|x|), \ff_{\varepsilon,k}(x))$,
  where the function $k$ is chosen large enough such that $m-k(\varepsilon,m,n) \le O(\log m \cdot \log (n/\varepsilon))$, 
  but not too large such that the triple can still be represented in binary using precisely $m$ bits.  
  On input a triple $(\varepsilon, k, y)$, the decompressor $\ttd'$ first enumerates the set $S$ of equation~\eqref{e:sets}, 
  and outputs $g_S(y)$, where $g$ is the monotone inverse of~$\ff_{\varepsilon,k}$.
\end{proof}

\begin{remark}\label{remark:inverse_lengths}
The theorem considers inputs $x$ of arbitrary length. 
However, it is enough to construct for each fixed input size $n$, 
an invertible $\ff \colon \{0,1\}^n \mapsto \{0,1\}^m$ with inverse~$g$. 
Indeed, by adding $n$ to the output, 
we obtain an inverse in $\{0,1\}^*$, simply by evaluating the inverse for $n$-bit inputs on the $n$-bit elements in~\eqref{e:sets}. 
\end{remark}

\begin{remark}\label{remark:effective_decompression}
The definition of $\Delta$-optimality 
requires the existence of some mapping $\ttd \mapsto \ttd'$, but does not impose any computability requirement on this mapping. 
However, from the proof of Theorem~\ref{th:main_single_source} we obtain that if $\ttd$ is given as an oracle, we can evaluate $\ttd'$ for an $n$-bit compressed string 
in space polynomial in $n$, and hence in time exponential in~$n$. Indeed, assume the length $n$ is a part of the output, as explained in the previous remark. On input $(\varepsilon, k, y, n)$, the decompressor $\ttd'$ uses the oracle to enumerate the set $S$ of $n$-bit strings in \eqref{e:sets}. Each time such a string is found, the value of $g_S(y)$ is calculated, and the output of $\ttd'$ is the corresponding value. The monotonicity property guarantees that the output of $\ttd'$ does not change after new elements appear in~$S$.
If $\ttd$ is computable in polynomial space, then so is the corresponding $\ttd'$. 
Similarly, if $\ttd$ is computable, respectively, partial computable, then so is $\ttd'$.
\end{remark}

\noindent
It remains to prove Theorem~\ref{th:main_invertible}. 
To better understand the main difficulty, we first review weaker results based on standard fingerprinting techniques.

%

\subsection*{Fingerprints from random hashing}

Fix $\varepsilon > 0$ and $K$. Theorem~\ref{th:main_invertible} provides $(K,\varepsilon)$-invertible functions with output length close to~$\log K$.
Consider the variant of Definition~\ref{def:invertible} in which the functions $\ff$ and $g$ are evaluated using shared randomness in their evaluations. In this model, $\ff$ and $g$ have an extra argument~$\rho$ representing a random string. 
Fingerprints given by random subset parities provide $(K,\varepsilon)$-invertible functions for this easier setting.

\medskip
\noindent
{\em The fingerprints.} 
On input $x$ and a sufficiently long~$\rho$, let $\ff_\rho(x)$ be a string of size $m = \lceil\log \tfrac{K}{\varepsilon}\rceil$ evaluated
by taking $m$ disjoint segments $\rho_1, \dots, \rho_m$ of length~$|x|$ from~$\rho$, and setting the $i$-th bit equal to $\sum_{j \le n} \rho_{i,j} x_j \bmod 2$.  In other words, $\ff_\rho(x) = Rx$, where $R$ is the matrix having the rows $\rho_1, \dots, \rho_m$,  and the $\rho_i$'s and $x$ are viewed as $n$-vectors with elements in the field GF[$2$]. 

\medskip
\noindent
{\em The inverse $g$.} 
On input $\rho$, a list $S$ of strings, and $y$, the value $g_{\rho,S}(y)$ is given by the first $z$ in $S$ for which $\ff_\rho(z) = y$.

\begin{lemma}\label{lem:subset_parity}
  If $S$ is a list of at most~$K$ strings of equal length and $x \in S$, then
  $
  \Pr_\rho \left[ g_{\rho,S}(\ff_\rho(x)) \mathop{=} x \right] \;\ge\; 1-\varepsilon.
  $
\end{lemma}

\begin{proof} 
  For two different strings $x$ and $z$ of equal length, the probability over $\rho$ that $\ff_{\rho}(x) = \ff_\rho(z)$
 is at most $2^{-m} \le \varepsilon /K$. 
 By the union bound, the probability that some $z \in S$ different from $x$, has the same subset parities, is at most~$\varepsilon$.
 This bounds the probability that $g$ returns an element different from~$x$.
\end{proof}

\bigskip
\noindent
How can shared randomness be eliminated?  One idea is to attach the random bits to the output of~$\ff$.
However,  on input $x$, the technique above uses  $m \cdot |x|$ random bits, 
and after attaching, the length exceeds the output length of the identity function.
We can try another well known hashing technique based on arithmetic modulo prime numbers, which only requires about $\log K$  random bits, 
and provides output lengths close to $2\log K$.
This hashing technique will be used later, and is based on the following.

\begin{lemma}\label{lem:simple_hashing}
 If $x, z_1, \dots, z_K$ are different nonnegative integers less than $2^n$ and $P$ 
 is a set of at least $Kn/\varepsilon$ prime numbers, 
 then for a fraction $1-\varepsilon$ of primes $p$ in $P$:
 $x \bmod p \;\not\in\; \{z_1 \bmod p, \dots, z_K \bmod p\}$.
\end{lemma}

\begin{proof}
 Each number $(x-z_i)$ has at most $n$ different prime factors. Thus all $(x-z_i)$'s together 
 have at most $Kn$ different prime factors. Hence a random  element in a set of at least
 $Kn/\varepsilon$ primes, is a prime factor with probability at most~$\varepsilon$.
\end{proof}

\noindent
Interpret strings $x$ as nonnegative integers smaller than $2^{|x|}$.   
Let the fingerprint $\fh_{\varepsilon,K}(x) = (p, x \bmod p)$ be given by selecting $p$ randomly among the primes of bit size at most $s + \lceil \log s\rceil + 1$, where $s = \log (K|x|/\varepsilon)$. 

\begin{lemma}\label{lem:prime_hashing}
  For all $\varepsilon > 0, K$ and $n$, the fingerprint $\fh_{\varepsilon,K}$ applied to $n$-bit strings, defines a $(K, \varepsilon)$-invertible function
  with output size at most $2\log K + O(\log \tfrac{n}{\varepsilon})$. 
\end{lemma}

\begin{proof}
  The prime number theorem implies that the $i$-th prime is less than~$2i \log i$ for large~$i$.
  Hence, the set of primes of bit size at most $s + \lceil \log s\rceil + 1$ contains at least $2^s = Kn/\varepsilon$ primes. 
  On input a set $S$ and a pair $(p,j)$, the function $g$ outputs an element $x$ in $S$ for which $x = j \bmod p$.
  By the lemma above, with probability $1-\varepsilon$, there exists no such $z$ different from~$x$, and hence, the decompressor outputs the correct value.
\end{proof}

\noindent
Let us summarize: a $(2^k,\varepsilon)$-invertible function has output size at least~$k$.
The fingerprints of Lemma~\ref{lem:subset_parity} have output size $k + O(\log \frac{1}{\varepsilon})$,  which is close to optimal, but use shared randomness. 
Lemma~\ref{lem:prime_hashing} does not use shared randomness, 
but the method cannot achieve compression lengths better than $2k + O(\log \frac{1}{\varepsilon})$.\footnote{
  An idea would be to improve   Lemma~\ref{lem:subset_parity} by computing  fingerprints with  fewer random bits. 
  A possibility would be to  use randomness generated by a pseudo-random generator.   
  In this case, one can prove a weaker variant of  Theorem~\ref{th:main_single_source} where one is restricted to 
  decompressors $\ttd$ running in polynomial time. 
  Moreover, the statement is  valid only conditional on a hardness assumption from computational complexity.  
   
  Another possibility would be to use Newman's theorem  from communication complexity, that converts protocols with shared randomness into protocols with private randomness. In this way, we can only achieve invertible functions for a fixed set $S$. This means that we obtain a variant of Theorem~\ref{th:main_single_source} where the optimality requirement considers only a single decompressor. Even if this decompressor runs in polynomial time, the resulting compressor runs in exponential time. If the decompressor is optimal, as in Corollary~\ref{cor:main_single_source}, the resulting compressor is not even computable.
  }

  The technique used in the above lemmas produces inverses $g$ of $\ff$ of a simple form: $g_S(y)$  searches for the first $x \in S$ such that $\ff(x) = y$ has positive probability. We show that this technique has a limitation that precludes optimal compression for all inputs. 
Let $F_x$ denote the set of all values of $\ff(x)$ that appear with positive probability. 
For the technique to work correctly, we need  that with positive probability $F_x$ is not included in $\bigcup \{F_z : z \in S, z \not= x\}$, in other words, the fingerprint $\ff(x)$ does not cause $x$ to collide with any other $z \in S$.
Unfortunately, such inverses for lists of suspects of size $2^k$, always require output sizes $m \ge 2k$, instead of the desired~$k$. 

\newcommand{\lemmaSimpleInverse}{
For all $x \in \mcX$, let $F_x \subseteq Y$. If $\#\mcY < \min\{\tfrac{1}{4} K^2, \# \mcX-\tfrac{1}{2}K\}$, 
  there exists an $x \in \mcX$ and a set $S \subseteq \mcX$ of size less than $K$ such that
  \[
  F_x \;\subseteq\; \bigcup_{ z \in S, z \not= x} F_z\,. 
  \]
  }
\begin{lemma}\label{lem:simple_inverse}
  \lemmaSimpleInverse
\end{lemma}

\noindent
The lemma follows from a more general theorem about union-free sets, see Jukna~\cite[Th 8.13]{juk:b:extremcombinatseconded}.
In appendix~\ref{sec:proof_simple_inverse} we give a simple proof of the lemma.

\subsection*{Fingerprints from condensers and conductors}

We bypass the issue from  Lemma~\ref{lem:simple_inverse} by allowing a limited amount of collisions. 
By using a second fingerprint and a more complex inverse function $g$, each input can still be recovered with high probability. 
The required first fingerprints are obtained from {condensers} and {conductors}, which have been studied in the theory of pseudo-randomness.  
We introduce the specific type of conductor that we use, after which we present an overview of the proof of Theorem~\ref{th:main_invertible}. 

Condensers have been introduced in various studies of extractors~\cite{raz-rei:c:extcon, rareva:c:extractor,ats-uma-zuc:j:expanders}, and
conductors  were explicitly introduced by Capalbo et al.~\cite{cap-rei-vad-wig:c:conductors}, under the name \emph{simple conductors}.  
The following closely related variant is tailored to our purposes. 
  Let $P$ be the probability measure of a random variable~$Y$ with finite range~$\mcY$. 
The  {\em $\gamma$-excess} of~$Y$ is 
   \[
     \sum_{y \in \mcY} \;\max \{0, P(y) - \gamma\} \;.
     \]

\noindent
For a set $S$, let $U_S$ be the random variable that is uniformly distributed in~$S$.

\begin{definition}\label{def:conductor}
  A probabilistic function $\ff \colon \mcX \rightarrow \mcY$ is a {\em \condenser{K}{K'}}
      if for every set $S \subseteq \mcX$  of size $\# S = K$, the $(1/K')$-excess of $\ff(U_S)$ is at most~$\varepsilon$.
      $\ff$ is a {\em $(K,\varepsilon)$-conductor} if it is a \condenser{K'}{K'} for all $K' \le K$.
\end{definition}

\noindent
Equivalently, $\ff$ is a $(K,\varepsilon)$-conductor if for every set $S$ of size at most $K$, 
the $(1/\#S)$-excess is at most~$\varepsilon$.
The equivalence with standard definitions in the literature is discussed in appendix~\ref{s:standard_def}.
The following lemma shows that every invertible function is a conductor. 
It is used in the proof of  Proposition~\ref{prop:lowerbounds}.

\begin{lemma}\label{lem:compressor_to_conductor}
  Every $(K,\varepsilon)$-invertible function is a $(K, \varepsilon)$-conductor.
\end{lemma}

\begin{proof}
  Let $\ff \colon \mcX \rightarrow \mcY$ be $(K,\varepsilon)$-invertible.
  For any $S \in \mcX^{\le K}$ and $x \in S$, we have 
  $g_S(\ff(x)) = x$ with probability at least~$1-\varepsilon$. 
  Thus $g_S( \ff(U_S)) \not= U_S$ with probability at most~$\varepsilon$. 

  Let $P$ be the measure of $g_S( \ff(U_S))$. 
  We show that this variable has $(1/\#S)$-excess at most~$\varepsilon$. 
  If we subtract from $P(x)$ the probability that $g_S( \ff(x)) \not= x$, the resulting 
  semimeasure has $(1/\#S)$-excess~$0$.
  Hence, $g_S(\ff(U_S))$ has $(1/\#S)$-excess at most~$\varepsilon$.

  Note that after applying a deterministic function, the excess can only increase. Since $y \mapsto g_S(y)$ 
  is deterministic, also $\ff(U_S)$ has excess at most~$\varepsilon$.
\end{proof}

\noindent
Remarkably, the relation in the other direction also holds true: from conductors we obtain invertible functions. This is the content of Theorem~\ref{th:main_invertible}. We present a sketch of the proof.
The main observation is that conductors produce fingerprints with some relaxed properties, but which are still good enough for decompression.
The first relaxation is that we do  not need to insist  on a fingerprint that produces zero collisions as in Lemma~\ref{lem:simple_inverse}. It is sufficient if, for each list of suspects  $S$, a string has a random fingerprint that causes at most $t$ collisions with other elements in~$S$, for $t$~polynomial or quasi-polynomial in~$n$.
Such a fingerprint is called a \emph{light} fingerprint.
This works, because to the light fingerprint, we simply append a second hash code  based on prime numbers.   By Lemma~\ref{lem:simple_hashing}, its size is~$O(\log (tn/\varepsilon))$, and this constitutes the overhead. Then the string can be isolated from the $t$ collisions, and be reconstructed.  The second relaxation is that we do not require that all the elements of $S$ have light fingerprints, but only at least half of them.  More precisely, we say that a string $x$ is \emph{deficient} if with probability $\varepsilon/2$, the fingerprint
$p= \ff(x)$ is not light, where $\ff$ is the conductor that we use as a hash function that produces fingerprints. Then, as we explain in the next paragraph, it suffices if at most half of the elements in $S$ are deficient, and a conductor $\ff$ indeed has this property.

Let $x$ be the string that we want to compress and let $p = \ff(x)$ be a random fingerprint of $x$ produced by the conductor $\ff$. We now sketch the decompression procedure which reconstructs $x$ from $p$ and the prime-based hash code.
Assume $S$ is  a valid set of suspects, i.e., $x \in S$.  Initially, we collect the first $t$ strings in $S$ that have $p$ as a fingerprint. If $x$ is non-deficient, $x$ will be among the collected strings with high probability.
But $x$ may very well be deficient, and, in this case, $x$ needs to be reanalyzed at a later stage. So, we form a smaller set 
$\fr(S)$ with all the deficient strings, and the decompressor is applied recursively to $\fr(S)$,  the new list of suspects. 
Since each recursive call decreases the set of suspects by half, the recursion has at most $\log (2\#S)$ levels, and we collect at most $t\log (2\#S)$ strings at all the levels of recursion. Now, the second prime-based hash code will allow us to distinguish $x$ among the collected strings and reconstruct it. The details are presented in the next section, where we prove Theorem~\ref{th:main_invertible}.

\bigskip

\noindent
We next present the condensers and conductors that we use in the proofs.
In our construction of $(2^k,\varepsilon)$-invertible functions, the difference between the output size and $k$ is proportional to the length of the second hash code, which is in turn proportional to the number of random bits used in the evaluation of the conductor.
In the pseudo-randomness literature, this is called the \emph{seed length}, and 
a $\Omega(\log \tfrac{n}{\varepsilon})$ lower bound has been proven,
see~\cite{nis-zuc:j:extract,rad-tas:j:extractors} and Proposition~\ref{prop:degreeLowerboundCondensers} below.
By a standard construction using the probabilistic method, there exist condensers that can be evaluated with $O(\log \tfrac{n}{\varepsilon})$ random bits.
They are computable, but unfortunately not computable in polynomial time.

\begin{proposition}\label{prop:existsPrefixExtractor}
  For all $\varepsilon$, $n$ and $k$, there exists a $(2^k,\varepsilon)$-conductor 
  $\ff \colon \{0,1\}^n \rightarrow \{0,1\}^{k+2}$ that uses  $\log \tfrac{4n}{\varepsilon}$ bits of randomness. 
\end{proposition}

\begin{remark*}
  Because conductors are finite objects, they can be computably constructed by exhaustive search.
  On input $\varepsilon, k$ and $n$, this search can be done in space exponential in $n$.
\end{remark*}

\noindent
The result of  Theorem~\ref{th:main_invertible} considers a polynomial time construction. 
A family of condensers or conductors is {\em explicit} if there exists a probabilistic algorithm that on every input, consisting of $\varepsilon$, $k$ and an $n$-bit $x$, outputs  $\ff^{\smallerbrackets{n}}_{\varepsilon, k}(x)$ in polynomial time. 

We obtain explicit families of conductors from known constructions of extractors.
Based on a result of Raz, Reingold, and Vadhan~\cite{rareva:j:extractor}, Capalbo et al. \cite{cap-rei-vad-wig:c:conductors} show that there exists an explicit conductor that uses $r= O(\log k \cdot \log^2 \tfrac{n}{\varepsilon})$ random bits.  In section~\ref{s:conductor_construction}, we obtain from explicit extractors given by Guruswami, Umans, and Vadhan~\cite{guv:j:extractor}, an improved explicit conductor that uses less randomness. 

\newcommand{\propExplicitConductors}{
  There exists an explicit family $\ff^{\smallerbrackets{n}}_{\varepsilon, k}$ 
  of $(2^k,\varepsilon)$-conductors $\ff^{\smallerbrackets{n}}_{\varepsilon, k} \colon \{0,1\}^n \rightarrow \{0,1\}^k$ that 
  uses $O(\log k \cdot \log \tfrac{n}{\varepsilon})$ random bits, for all $\varepsilon$, $k$ and $n$. 
  }

\begin{proposition}\label{prop:conductors1}
  \propExplicitConductors
\end{proposition}

\noindent
We also present a closely related, but technically less cumbersome approach, which does not use conductors, but instead obtains invertible functions by combining a few condensers given by the following result.

\newcommand{\lemmaGUVpartial}{
  For all $n,\kappa$ and $\varepsilon$, 
  there is an explicit family of \condenser{2^{2\kappa}}{2^\kappa}s $\ff \colon \{0,1\}^n \rightarrow \{0,1\}^\kappa$ 
  whose evaluation requires $O(\log \tfrac{n}{\varepsilon})$ bits of randomness.
  }

\begin{theorem}[\cite{guv:j:extractor}, Theorem 1.5 or 4.17]\label{th:GUV_partial}
  \lemmaGUVpartial
\end{theorem}

\begin{remark*}\label{r:extractor_condenser}
  The cited theorem in~\cite{guv:j:extractor} gives an extractor 
  with output length  $m \ge k/2$.  We can set $m = \ceil(k/2)$, 
  because in extractors, we can reduce the output size by merging equal numbers of outputs, 
  and this does not increase the statistical distance to the uniform measure. 
 Finally, we obtain a condenser, since a function is a $(2^{2\kappa},\varepsilon)$-extractor if and only if it is a~\condenser{2^{2\kappa}}{2^\kappa}, see Remark~\ref{remark:extractor_vs_condenser} in appendix~\ref{s:standard_def}.
\end{remark*}

\subsection*{Proof of Proposition~\ref{prop:single_source_nonpoly} and Theorem~\ref{th:main_invertible}}\label{s:proof_single}

The following corresponds to the weak invertibility requirement in the above proof sketch. 

\begin{definition*}\label{def:invertible_with_rejections}
  $\ff \colon \mcX \rightarrow \mcY$ is online {\em $(K,t,\varepsilon)$-list invertible} with $T$ {\em rejections} 
  if there exist deterministic monotone functions $\fr \colon \mcX^{\le K} \rightarrow \mcX^{\le T}$ 
  and $\fg \colon \mcX^{\le K} \times \mcY \rightarrow \mcX^{\le t}$ such that for all $S \in \mcX^{\le K}$ 
  and all $x \in S \setminus \fr(S)$:
  \[
    \Pr \left[ x\in \fg(S,\ff(x)) \right] \;\ge\; 1-\varepsilon.
  \]
\end{definition*}

\noindent
The interpretation is that $G$ is the ``pruning" function, which on inputs $S$ 
and $y$,  reduces the list of $K$ suspects to a smaller list of at most $t$ elements. 
For $y = F(x)$, this short list should contain $x$, provided that $x \in S$ and $x$ does not belong to the set $\fr(S)$ of at most $T$ ``deficient elements."  $\fr$ is the  ``reanalyze'' function that determines the elements in $S$ that can be incorrectly lost in the pruning, and need to be reanalyzed again.  

\begin{lemma}
  \label{lem:condenser_as_filter}
  Let $a$ be an integer.
  Every \condenser{K}{\tfrac{1}{a}K} that can be evaluated with $r$ random bits is 
  online $(K,a2^r,2\varepsilon)$-list invertible with $K/2$ rejections.
\end{lemma}

\noindent
In this section we use the lemma with~$a=1$.

\begin{proof} 
  For $S \in \mcX^{\le K}$ and $y \in \mcY$, let $\fg(S,y)$ be the list containing the first $a2^r$ appearances of elements $x$ in $S$ for which 
  $\ff(x) = y$ has positive probability  (or all appearances, if there are less than $a2^r$ of them).

  Note that $\fg$ is  monotone in $S$.
  We assume that $S$ has size exactly $K$, since the invertibility conditions only become easier to prove for smaller~$S$. 
  For a list $L$, let $U_L$ denote the random variable obtained by selecting a random element in~$L$.
  Observation: 
  \begin{quote}
    \textit{If $L'$ is the sublist of some $L \in \mcY^B$ obtained by retaining the first $b$ appearances of each element $y \in L$, 
  then $\Pr[U_L \not\in L']$ is at most the $(b/B)$-excess of $U_L$.
  }
  \end{quote}
  
  \noindent
  We \st{first} show that
  \[
    \Pr \left[ U_S \not\in \fg(S,\ff(U_S)) \right] \;\le\; \varepsilon.
  \]
  Let $L$ be the list of length $2^r K$ of the values $\ff_\rho(x)$ for all $x \in S$ and all assignments $\rho$ of~$r$ random bits in $\ff(x)$.
  More precisely, for $S = [x_1, \dots, x_K]$, concatenate $[\ff_{0^r}(x_i), \dots, \ff_{1^r}(x_i)]$ for increasing values of~$i$.
  Note that $U_L$ has the same distribution as $\ff(U_S)$, which 
  by definition of condenser, has $\tfrac{a}{K}$-excess at most~$\varepsilon$.
  The procedure for $\fg$ defines a sublist $L'$ of $L$, containing all first $a2^r$ appearances of some~$y \in \mcY$.
  By choice of $L$ and $L'$, the events $U_L \not\in L'$ and $U_S \not\in \fg(S,\ff(U_S))$ have precisely the same probability.
  The inequality follows by applying the above observation, with $b = a 2^r, B = K2^r$.

  The function $\fr$ is obtained by selecting the elements $x \in S$ for which $\Pr \left[x \not\in \fg(S,\ff(x)) \right] \,>\, 2\varepsilon$.
  $\fr$ is monotone, because after adding an element to $S$, these probabilities for its other elements do not change.
  The proof finishes by showing that $\fr(S)$ contains at most $K/2$ elements, i.e.,  
  $\Pr[U_S \in \fr(S)] \le 1/2$. This follows from
  \[
   2\varepsilon \cdot \Pr [U_S \in \fr(S)] \,\;\le\;\, \Pr [ U_S \not\in \fg(S,\ff(U_S))] \,\;\le\;\, \varepsilon.
   \qedhere
  \]
\end{proof}

\begin{remark*}
  If the function $\ff \colon \{0,1\}^n \rightarrow \mcY$ can be evaluated in space polynomial in $n$, 
  then also the functions $\fr$ and $\fg$ constructed in the proof above can be evaluated in space polynomial in~$n$.
\end{remark*}

\begin{corollary}\label{cor:pruning}
  If $\ff$ is a $(M,\varepsilon)$-conductor that can be evaluated with $r$ random bits, 
  then $\ff$ is online $(M,2^r \log (2M), 2 \varepsilon)$-list invertible (with $0$ rejections). 
\end{corollary}

\begin{proof}
  Let $\fg$ and $\fr$ be the functions defined above for~$a=1$. Note that
  we have $\# \fr(S) \le \tfrac{1}{2}\# S$ for all $S$, since the algorithms do not depend on $K$, 
  and the assumption holds for all $K \le M$, by definition of conductors.

  For $S \in \mcX^{\le M}$ and $y \in \mcY$, the inverse $\fg'$ that satisfies the conditions is defined recursively. 
  If $S$ is empty, then $G'(S,y)$ is empty. Otherwise, $G'(S,y)$ is the concatenation of
  $\fg(S,y)$ and $\fg'(\fr(S),y)$.  

  Each recursive call adds at most $2^r$ elements, and at most $\log (2M)$ recursive calls are made.  
  The probability that $x \not\in \fg'(S,\ff(x))$ is at most $2\varepsilon$. Indeed, 
  if $x \not\in \fr(S)$,  this follows from Lemma~\ref{lem:condenser_as_filter},  
  and otherwise, this follows by an inductive argument on the size of~$S$.
\end{proof}

\begin{remark}\label{remark:inverse_polynomial_space}
  If the function $\ff \colon \{0,1\}^n \rightarrow \mcY$ can be evaluated in space polynomial in $n$, 
  then also the functions $\fr'$ and $\fg'$ constructed in the proof above can be evaluated in space polynomial in~$n$.
  Indeed, these functions operate on sets of $n$-bit strings and can be of exponential size. But, 
  to evaluate the functions $\fr(S)$ and $\fg(S,y)$, one does not need to store the full set $S$,
  but only needs to iterate over elements in~$S$. 
  This implies that the space needed to evaluate $\fg'(S,y)$ 
  equals the recursion depth times the space needed to iterate over $\fr(\cdot)$ and $\fg(\cdot, y)$, 
  and this is polynomial in~$n$. 
\end{remark}

\noindent
The next result follows almost directly from the definitions.

\begin{lemma}\label{lem:compose_listInvertible_and_invertible}
  If $\ff$ is $(K,t,\varepsilon)$-list invertible with $T$ rejections, $\ff'$ is $(t\mathop{+}T,\varepsilon')$-invertible, and both functions have the same domain,
  then $x \mapsto (\ff(x), \ff'(x))$ is $(K,\varepsilon \mathop{+}\varepsilon')$-invertible.
\end{lemma}

\noindent
The next result follows directly from this lemma and Corollary~\ref{cor:pruning}. 

\begin{corollary}\label{cor:conductor_invertible}
  If $\ff \colon \{0,1\}^n \mapsto \mcY$ is a $(K,\varepsilon)$-conductor that can be evaluated using at most $r$ random bits, 
  then the function \[x \;\;\longmapsto\;\; (\ff(x), \fh_{\varepsilon,s}(x)) \quad \textnormal{ with } \quad s = 2^r \log (2K)\] is $(K,3\varepsilon)$-invertible, where $ \fh_{\varepsilon,s}$ is the prime-based hash function from Lemma~\ref{lem:prime_hashing}.
\end{corollary}


\begin{proof}[Proof of  Theorem~\ref{th:main_invertible}.]
  We apply Corollary~\ref{cor:conductor_invertible} to the conductor given in Proposition~\ref{prop:conductors1}. 
  The result is a family of online $(2^k,\varepsilon)$-invertible functions where 
  $\ff \colon \{0,1\}^n \rightarrow \{0,1\}^m$ with $m-k \le O(\log k \cdot \log \tfrac{n}{\varepsilon})$.
  By Remark~\ref{remark:inverse_polynomial_space}, the inverse can be evaluated in space polynomial in~$n$.
  Finally, by Remark~\ref{remark:inverse_lengths}, we obtain an inverse function over the set of strings of all~lengths.
\end{proof}

\begin{proof}[Proof of  Proposition~\ref{prop:single_source_nonpoly}.]
  We apply  Corollary~\ref{cor:conductor_invertible} to the conductor given in  Proposition~\ref{prop:existsPrefixExtractor}. 
  The result is a family of online $(K,\varepsilon)$-invertible functions where $\# \mcY/K$ is polynomial in~$n/\varepsilon$.
  Proposition~\ref{prop:single_source_nonpoly} follows by converting invertible functions to compressors, 
  as explained in the proof of Theorem ~\ref{th:main_single_source} assuming Theorem~\ref{th:main_invertible} presented at the beginning of section~\ref{s:techniques}.
\end{proof}

\subsection*{Proof of Theorem~\ref{th:main_invertible} based on condensers}

The proof of  Theorem~\ref{th:main_invertible} given in the previous section requires the explicit conductor  from Proposition~\ref{prop:conductors1}, whose proof  is rather technical.  In this subsection, we  give an easier and more direct  proof that is based only on condensers, bypassing conductors. The two  proofs are similar, 
still, they each have their own advantage:
if explicit conductors are discovered that use less randomness, then this leads to improved compressors through the first proof. 
On the other hand, if explicit condensers are found that use logarithmic randomness and extract more of their minentropy, then the approach of this section 
leads to improved compressors.
The following is a consequence of Lemma~\ref{lem:condenser_as_filter}. 

\begin{lemma}\label{lem:pruning_large}
  Let $a$ and $b$ be integers. If $\ff$ is a \condenser{M}{\tfrac{1}{a}M} that can be evaluated with $r$ random bits, 
  then $\ff$ is online $(bM,ab2^r \log (2b),2\varepsilon)$-list invertible with $M/2$ rejections.
\end{lemma}

\begin{proof}
  The pruning and reanalyze functions $\fg'$ and $\fr'$ are defined recursively in~$b$.
  Given $S \in \mcX^{\le bM}$, partition $S$ in $b$ sublists $S'$ of length at most~$M$, 
  and apply the functions $\fg$ and $\fr$ from  Lemma~\ref{lem:condenser_as_filter}.
  Collect the  to-be-reanalyzed strings from all $\fr(S')$'s in a set $S_\textnormal{rec}$, and the selected 
  strings from all sets $\fg(S',y)$ in a set $S_\textnormal{sel}$.
  If $\# S_\textnormal{rec} \le M/2$ we are done, and we let $\fr'(S) = S_\textnormal{rec}$ and $\fg'(S,y) = S_\textnormal{sel}$.
  Otherwise, we continue recursively by letting $\fg'(S,y)$ be the concatenation of 
  $S_{\textnormal{sel}}$ and $\fg'(S_\textnormal{rec},y)$, and letting $\fr'(S) = \fr'(S_\textnormal{rec})$. 

  We show that the algorithm works correctly. By construction, $\fr'$ 
  outputs the required number of elements.
  In each recursive call, $\# S_\textnormal{rec}$ is at least halved.
  Thus, the number of recursive calls is at most $\log (2b)$. 
  In each call, the size of $S_{\textnormal{sel}}$ is at most $b \cdot (a2^r)$, and hence this number 
  of elements is appended to $\fg'$. 
  Thus, this function outputs a list of size at most $ab2^r \log (2b)$.
\end{proof}

\noindent
Using Lemma~\ref{lem:compose_listInvertible_and_invertible}, this implies the following for $a = b = K$ and $M = K^2$. 

\begin{lemma}\label{lem:compose_condenser_invertible}
  If $\ff$ is a \condenser{K^2}{K} that can be evaluated with $r$ random bits, $\ff'$ is $(K^2 2^r \log (4K), \varepsilon')$-invertible, 
  and both functions have the same domain, 
  then $x \mapsto (\ff(x), \ff'(x))$ is $(K^3,2\varepsilon+\varepsilon')$-invertible.
\end{lemma}

%

\noindent
We apply this lemma to the following condensers given by Guruswami, Umans, and Vadhan~\cite{guv:j:extractor}. 

\begin{theorem*}[Restated]
  \lemmaGUVpartial
\end{theorem*}

\begin{proof}[Proof of Theorem~\ref{th:main_invertible}.]
  By Remark~\ref{remark:inverse_lengths}, it suffices to present the construction for inputs of a fixed length~$n$. 
  If $k \ge n$, then let $\ff$ be the identity function. 
  Otherwise, we apply Lemma~\ref{lem:compose_condenser_invertible} recursively for decreasing~$k$. 
  In other words, we concatenate $O(\log k)$ condensers with geometrically decreasing values of~$\kappa$. 

  We present the details. 
  Let $r$ denote the $O(\log \tfrac{n}{\varepsilon})$ bound on the randomness in the condenser of Theorem~\ref{th:GUV_partial}.  
  Let $b$ be 
  such that $2^b \ge 2^r \log (4K)$ for all $K \le 2^n$, but still satisfies $b \le O(\log \tfrac{n}{\varepsilon})$. 

  If $k < 100 \cdot b$, we obtain an $(2^k,\varepsilon)$-invertible function using prime hashing from Lemma~\ref{lem:prime_hashing}.
  Otherwise, the invertible function is obtained by concatenating the condenser for $\kappa = \ceil(k/3)$ 
  and the recursive application of the construction for~$k \leftarrow 2\kappa + b$. 
  
  At most $d \le O(\log k)$ recursive calls are made, because if $k \ge 100 \cdot b$, 
  the recursive value for $k$ is close to $\tfrac{2}{3}k$, say at most $\tfrac{5}{6} k$.
  We verify that if the recursion depth is $d$, then
  we obtain a $(2^k, 2d\varepsilon)$-invertible function with output length $k + bd$.
  Indeed, in each concatenation, the error given by Lemma~\ref{lem:compose_condenser_invertible} 
  increases with $2\varepsilon$, and the output length increases with $\kappa$, while the value of $k$ increases by $\kappa - b$.
  This proves the first part of the theorem.
  The moreover part follows from similar observations as in Remark~\ref{remark:inverse_polynomial_space}.
\end{proof}

\section{Optimal distributed compression: Proof of Theorem~\ref{th:main_distributed_compression}}\label{s:proof_distributed}

Observe that if $\ff_1 \colon \mcX \rightarrow \mcY_1$ and $\ff_2 \colon \mcX \rightarrow \mcY_2$ 
are $(K_1, \varepsilon_1)$ and $(K_2, \varepsilon_2)$-invertible, 
then $x \mapsto (\ff_1(x), \ff_2(x))$ is $(K_1K_2, \varepsilon_1\mathop{+}\varepsilon_2)$-invertible. 
We can not apply this observation to the setting of distributed compression. Indeed, if we choose $\mcX = \mcX_1 \times \mcX_2$, 
and consider a function $\ff_1$ acting on the left and $\ff_2$ on the right coordinate,
then $\ff_1$ can not be invertible on $\mcX$, because it can not distinguish tuples with the same right  coordinate. 
The following proposition obtains invertible functions suitable for distributed compression with 2 sources.

\newcommand{\propInvertibleTwoFunctions}[1]{ 
 If $\ff_1\colon \mcX_1 \rightarrow \mcY_1$ is $(K_1, \varepsilon)$-invertible 
 and $\ff_2 \colon \mcX_2 \rightarrow \mcY_2$ is $(K_2, \varepsilon)$-invertible, then 
 $(x_1, x_2) \mapsto (\ff_1(x_1), \ff_2(x_2))$ is $(3\varepsilon)$-invertible#1
 in sets $S \subseteq \mcX_1 \times \mcX_2$ 
 for which
\begin{gather*}
 \# S \;\le\; K_1K_2, \\
 \# \left\{ (z_1, z_2) \in S : z_2 = x_2  \right\} \;\le\; K_1 \qquad \forall x_2 \in \mcX_2,\\
 \# \left\{ (z_1, z_2) \in S : z_1 = x_1  \right\} \;\le\; K_2 \qquad \forall x_1 \in \mcX_1.
 \end{gather*}
}

\begin{proposition}\label{prop:invertible_two_functions}
  \propInvertibleTwoFunctions{\footnote{
    We say that $\ff$ is invertible in a set $S$ if there exists a function $h$ such that 
    $\Pr[h(\ff(x)) = x] \ge 1-\varepsilon$ for all~$x \in S$.
    }
 }
\end{proposition}

\noindent
The proof is given in appendix~\ref{s:sw-easy}.
For the proof of Theorem~\ref{th:main_distributed_compression}, 
we need an online variant of this proposition for an arbitrary number of invertible functions.
Its statement requires some more definitions. 

\begin{definition}\label{d:kbound}\label{def:small_slices}
  Let $S \subseteq \mcX_1 \cdots \mcX_\ell$ and $\tK  = (K_1 \dots K_\ell)$ be an $\ell$-tuple of integers.
  $S$ has {\em $\tK$-small slices} if for all $\tx \in S$ and for all  $J \subseteq [\ell]$:
  \[
  \# \{\tz \mathop{\in} S : \tz_J \mathop{=} \tx_J \} \,\;\le\;\,  \prod_{\!\!j \in [\ell] {\setminus} J\!\!} K_j .
  \]
\end{definition}


\noindent
Let $\mcX^{<\omega}$ denote the set of finite sequences of elements in $\mcX$.

\begin{definition}\label{def:distribinvertible}
  Let $\mcS$ be a collection of subsets of~$\mcX$.
  $\ff$ is $\varepsilon$-invertible {\em on $\mcS$} if there exists a mapping $g\colon \mcX^{<\omega} \times \mcY \rightarrow \mcX$,  
  such that for all $S$, whose elements form a set in $\mcS$
 \[
   \Pr \left[g_S(\ff(x)) = x\right] \;\ge\; 1-\varepsilon\,,
 \]
 $\ff$ is {\em online} $\varepsilon$-invertible on $\mcS$ if there exists a monotone such~$g$, as in Definition~\ref{def:invertible}.
\end{definition}

\noindent
Theorem~\ref{th:main_distributed_compression} follows from the proof of the following.

\begin{theorem}\label{th:main_componentwise_invertible}
  If $\ff_j \colon \mcX_j \rightarrow \mcY_j$ is online $(\tfrac{\ell}{\varepsilon}K_j, \varepsilon)$-invertible for all $j \in [\ell]$, 
  then the function 
  \[
   (x_1, \ldots, x_\ell)\;\; \longmapsto\;\; (\ff_1(x_1), \ldots, \ff_\ell(x_\ell)) 
  \]
  is online $(8\ell\varepsilon)$-invertible on the collection of sets with $\tK$-small slices.
\end{theorem}

\begin{remark*}
  The offline variant also holds, in which we start with a weaker assumption and obtain a weaker conclusion. 
\end{remark*}

\noindent
The proof of Proposition~\ref{prop:invertible_two_functions} can be modified for the online version and extended to arbitrary $\ell$. 
Moreover, compared to Theorem~\ref{th:main_componentwise_invertible} it is easier, 
and provides slightly better parameters when~$\ell$ and~$\varepsilon$ are small.
Unfortunately, the obtained error is~$\ell 2^\ell\varepsilon$,  which can be much larger than~$8\ell\varepsilon$. 

The difficulty in proving both the proposition and the theorem is to partition the large set $S$ of initial suspects 
into smaller subsets on which the inverses of the functions $\ff_j$ can be applied.
For the proposition, this partition is deterministic and rather intuitive. 
For the theorem, we use a different and novel technique, called \emph{randomized-tree partition}, and 
this allows to reduce the error to~$O(\varepsilon \ell)$.

\subsection*{Proof of Theorem~\ref{th:main_componentwise_invertible}}

We construct a probabilistic algorithm $g$ that on input $\ty \in \mcY$ and a list $S$ of tuples in $\mcX_1 \times \cdots \times \mcX_\ell$  
with $\tK$-small slices, selects an element $g(S,\ty)$  from $S$ such that for each fixed $\tx \in S$: 
\begin{equation*}\tag{*}\label{eq:error_prob}
  \Pr \left[ g(S,\ff(\tx)) \not= \tx \right] \; \le \; (\exp(1)+\ell) \varepsilon,
\end{equation*}
in which the probability is over the random choices of both $g$ and~$\ff$. 
The algorithm operates in an {\em online} way for the input~$S$. 
This means that $g$ starts by executing an initialization procedure that does not depend on~$S$.
Next, for each subsequent element $\tx$ in~$S$, an update procedure is run that manipulates some variables.
This update procedure might decide to assign the output $g(S, \ty) = \tx$. 
Once, committed to an output, this output may no longer be changed. 
In this way, the obtained function $g$ is monotone.

A deterministic online inverse is obtained by taking majority votes of the probabilistic outcome of~$g$.  
This  at most doubles the error. 
Indeed, given a set $S$, a value $\ty$ is {\em bad} for some $\tx \in S$ if $g(S, \ty) = \tx$ has probability at most $1/2$. 
By the inequality above, the output $\ff(\tx)$ is bad for $x$ with probability at most $2 \cdot (\exp(1)+\ell) \varepsilon \le 8 \ell \varepsilon$, 
and if $\ff(\tx)$ is not bad, then the majority vote returns~$\tx$.
Hence, the derandomized $g$ satisfies the invertibility conditions of the theorem, 
where the correspondence with the notation from the definition is obtained by~$g_S(\ty) = g(S,\ty)$.

The initialization algorithm for $g$  creates a random tree. 
It starts by creating the root, which has depth~$0$. 
Afterwards, for $j = 1,\dots, \ell$ and for each node at depth $j-1$, 
it creates $\tfrac{\ell}{\varepsilon} K_j$ children and associates each element in~$\mcX_j$ to a random child.
After the random tree is created, the initialization procedure finishes.

Recall that the update procedure for $g(S, \ty)$ is run when a new element $\tx$ is added to~$S$. 
This procedure is deterministic and places copies of $\tx$ on some nodes of the tree. 
We can view the copies of $\tx$ as pebbles labeled by~$\tx$. 
Elements  are never removed from nodes, and we ensure that at each moment, there is at most one element on a node. 
The update procedure proceeds in two steps, which we call  \emph{top-down percolation} and \emph{bubble up}.

\medskip
\noindent
{\em The top-down percolation step.} 
\\Place $\tx$ on a leaf of the tree by descending it in the natural way: 
for moving down from depth $j-1$ to depth $j$, select
the child node associated to $x_j$. 
If the leaf already contains an element, then $\tx$ is not placed on any leaf, 
and the update procedure terminates without selecting $\tx$ for the output.

\medskip
\noindent
Before presenting the second step, we first bound the probability that a fixed $\tx$ in $S$ 
is not placed on a leaf during the percolation step.

\begin{lemma}\label{lem:random_tree}
  Assume $S$ has $\tK$-small slices and that $\tx$ belongs to $S$. 
  The probability that $\tx$ is not placed on a leaf is at most $\varepsilon \exp(1)$.
\end{lemma}

\begin{proof}
  We bound the probability that some tuple $\tz$ in $S$, different from $\tx$,  descends to the same leaf as~$\tx$. 
  If this is not the case, than $\tx$ is indeed placed on a leaf.

  Let $J$ be a non-empty subset of~$[\ell]$ and let $\tz$ be an $\ell$-tuple for which $z_j \not= x_j$ for all $j \in J$, 
  and $z_j = x_j$ for all $j \in [\ell] \setminus J$. 
  The probability that $\tz$ is placed in the same leaf as $\tx$ is at most 
 \[
   \prod_{j \in J} \frac{\varepsilon}{\ell K_j}.
 \]
 By the slice assumption, the set $S$ contains at most $\prod_{j \in J} K_j$ such elements $\tz$,
 and hence, the probability that this happens for such a $\tz$ is at most 
 \[
 (\varepsilon/\ell)^{\# J}  \;\le\; \varepsilon (1/\ell)^{\# J},
 \] 
 since $J$ is nonempty.
 By the union bound over all nonempty subsets $J \subseteq [\ell]$, the probability that some element of $S$ is placed in $\tx$'es leaf is at most
 \[
 \le \;\; \varepsilon \sum_J \left(\frac{1}{\ell}\right)^{\# J}
 \;\; \le \;\; \varepsilon \left(1+ \frac{1}{\ell} \right)^\ell 
 \;\;\le\;\; \varepsilon  \exp(1).
 \qedhere
 \]
\end{proof}

\noindent
The second step of the update procedure is the {\em bubble up} step. 
In this step, the element $\tx$ assigned to some leaf,  attempts to be copied upwards to the root. 
There is a ``competition,'' and out of the elements placed on the children of a node, at most one progresses up. 
If $\tx$ reaches the root, then it is the output of~$g(S, \ty)$.
Let $g_j$ be the online function that inverts~$F_j$.

\medskip
\begin{samepage}
\noindent
{\em The bubble up step.}
\\ Repeat the rounds that we describe next until termination. At the current round, execute the following four steps:
\begin{itemize}[leftmargin=*, itemsep=3pt, topsep=3pt, parsep=0pt, label=-]
  \item 
    Consider the node with minimal depth that contains $\tx$ and let $j$ be this depth  (so, at the first round,  $j$ is $\ell$).
  \item 
    Let $B_j$ be the set of $j$-th coordinates of tuples placed on  the siblings of this node. We view these tuples as competing to create a copy of themselves in the parent.
  \item 
    If $g_j(B_j, y_j) = x_j$,  place a copy of $\tx$ on the parent node (in other words, $\tx$ is the winner).  Otherwise, terminate without selecting~$\tx$.
  \item 
    If $\tx$ is placed in the root, terminate and select $\tx$ for the output.
\end{itemize}
\end{samepage}

\medskip
\noindent
{\em Remarks.}

The above loop always terminates, 
because if termination does not happen in the third step of a round, 
$\tx$~is raised up in the tree and the algorithm continues with the next round.  
Therefore, if termination does not happen in a third step of any round, 
$\tx$~reaches the root and the algorithm terminates in the fourth step. 

As promised above, in every node, at most one tuple is placed. Indeed, for leaves this follows by the  construction in the 
percolation step. For the other nodes, this follows by the  definition of online inverses $g_j$, because if $g_j(B_j, y_j)$ has committed to an output $x_j$, then further
additions of elements to $B_j$ will not change this output. By the construction of the tree, 
all siblings in the third step contain tuples with different $j$-th coordinates, and therefore only at most one can win. 

When the algorithm commits to an output, this answers remains. Indeed, by the previous 
step, at most one tuple is placed in the root, and after being placed, it is never removed.

\medskip
\noindent
We prove~\eqref{eq:error_prob}. For a tuple $\tx \in S$, consider  the following events:
\\- ${\cal E}_0$  is the event that $\tx$ is assigned to a leaf of the tree at the percolation step.
\\- For each $j \in [\ell]$, consider the node at depth $j$ on the branch of $\tx$. 
Let $B_j$ be the set of $j$-th coordinates of tuples placed in the node's siblings  after processing all elements of~$S$.
 ${\cal E}_j$  is the event that $g_j(B_j, F_j(x_j)) = x_j$.

\medskip
\noindent
Conditioned on all events ${\cal E}_0$ and ${\cal E}_\ell, \ldots, {\cal E}_1$ being true, $\tx$ will appear in the root and, thus, is the output of $g(S, \ff(\tx))$, in other words,  $g$ has correctly inverted $\ff(\tx)$. Indeed, by ${\cal E}_0$, $\tx$~is placed to some leaf during the percolation step.
And by ${\cal E}_j$, at each depth $j$ in the bubble up step, the tuple $\tx$ wins the competition and is placed on its parent, 
since, no other sibling has the same $j$-th coordinate.

We now bound the error probability that at least one of the above events does not happen. 
By Lemma~\ref{lem:random_tree}, the probability that $\tx$ is not placed on a leaf at the percolation phase is at most~$\varepsilon \exp(1)$.
Otherwise, $\tx$ is placed to some leaf. Consider its branch. Note that  the size of $B_\ell$ is at most $\tfrac{\ell}{\varepsilon}K_\ell$, because this is the number of sibling nodes, and each node contains  at most one tuple.
By the $(\tfrac{\ell}{\varepsilon}K_\ell,\varepsilon)$-invertibility of~$\ff_\ell$,  
the value of $g_\ell (B_\ell,  \ff_\ell(\tx_\ell))$ differs from $\tx_\ell$ with probability at most $\varepsilon$. 
Otherwise, if this does not happen, then $\tx$ is placed in the branch at depth $\ell - 1$, and $x_{\ell-1}$ belongs to $B_{\ell-1}$. 
By induction, the same argument is valid for all levels, and thus every event ${\cal E}_j$ fails with probability at most~$\varepsilon$.

The probability that one of the events does not happen is bounded by the sum of all mentioned probabilities, 
which is precisely the right-hand side of~\eqref{eq:error_prob}. 
The theorem is proven.


\subsection*{Proof of  Theorem~\ref{th:main_distributed_compression}.}

We first connect the Slepian-Wolf constraints from Definition~\ref{def:SW_rate} to $K$-small slices. 

\begin{lemma}\label{lem:fromKolm_to_smallSlices}
The set $S$ of tuples $\tx$ such that $\tk$ satisfies
the $\ttd$-Slepian-Wolf constraints for $\tx$ has $(2^{k_1}, \dots, 2^{k_\ell})$-small slices. 
\end{lemma}

\begin{proof}
  If $J = [\ell]$, the set in  Definition~\ref{def:small_slices} equals $\{\tx\}$ and the inequality is true.
  Fix a strict subset $J \subset [\ell]$. 
  All elements in $\tz \in S$ satisfy 
  \[
    \C_\ttd(\tz_{[\ell] \setminus J} \,|\, \tz_J) \; <\, \sum_{j \in [\ell] \setminus J} k_j.
  \]
  The number of such $\ell$-tuples $\tz$ with $\tz_J = \tx_J$ is therefore bounded by the exponent of the 
  right hand side.  Hence, this set has the required size. 
\end{proof}

\noindent
Recall that in the proof of Theorem~\ref{th:main_single_source}, 
we constructed a compressor using a family $\ff_{\varepsilon,\kappa}$ of $(2^\kappa, \varepsilon)$-invertible functions on~$\{0,1\}^*$. 
For this, we simply appended the parameters $\varepsilon$ and $\kappa$ to the output. Note that for the obtained overhead $\Delta$, 
the output length is equal to~$m = \kappa + \Delta - 1$.
For such compressors, Theorem~\ref{th:main_distributed_compression} follows easily from Theorem~\ref{th:main_componentwise_invertible}.

\begin{proof}[Proof of Theorem~\ref{th:main_distributed_compression} for compressors that have the parameter $\kappa$ in their output.]
  Let $\varepsilon$ and $\ty = (y_1, \dots, y_\ell)$ be inputs for the decompressor $\ttd'$ that we need to construct.
  By the  assumption, each $y_j$ provides us with the value $\kappa_j$, and  therefore specifies  a monotone $(2^{\kappa_j}, \varepsilon)$-inverse~$g_j$.
  Let $g$ be the inverse constructed in the proof of Theorem~\ref{th:main_componentwise_invertible}. This construction only uses objects that 
  are specified by the inputs $\varepsilon$ and $\ty$ of $\ttd'$: $\varepsilon$, $\ell$, $\kappa_j$ and $g_j$ for all~$j$. 
  Let $S$ be the set of strings for which $(\kappa_1 - \log \tfrac{\ell}{\varepsilon}, \dots, \kappa_\ell - \log \tfrac{\ell}{\varepsilon})$ 
  satisfies the $\ttd$-Slepian-Wolf constraints.
  Define $\ttd'(\varepsilon,\ty) = g(S,\ty)$.

  If $\tx$ satisfies the 
  $(m_1 - \Delta_1 - \log \tfrac{\ell}{\varepsilon}, \dots, m_\ell - \Delta_\ell - \log \tfrac{\ell}{\varepsilon})$-Slepian-Wolf conditions, 
  then $\tx \in S$ by choice of $S$, (there is even 1 bit surplus).   
  By  Lemma~\ref{lem:fromKolm_to_smallSlices}, 
  $S$ has $(\tfrac{\varepsilon}{\ell}2^{\kappa_1}, \dots, \tfrac{\varepsilon}{\ell}2^{\kappa_\ell})$-small slices. 
  By choice of $\ttd'$ and by Theorem~\ref{th:main_invertible}, 
  the event $\ttd'(\varepsilon,\ff(\tx)) = \tx$ has probability at least 
  $1-8\varepsilon \ell$. This implies the theorem. 
\end{proof}

\noindent
For the general case, we need to adapt the proof of Theorem~\ref{th:main_componentwise_invertible}.

\begin{proof}[Proof of  Theorem~\ref{th:main_distributed_compression}.] 
  Assume the compressor is $\Delta$-optimal, and for inputs $(\varepsilon,m,x)$, 
  let $\Delta(\varepsilon, m, x)$ be the value of the overhead.
  Given a decompressor $\ttd$ and an $\ell$-tuple $\tm$, 
  we define the set $S_\tm$ of tuples $\tx$ that satisfy the conditions of the theorem:
  \[
  C_\ttd(\tx_J \mid \tx_{[\ell] \setminus J}) \; \le \; \sum_{j \in J} \left(m_j - \Delta(\varepsilon,m_j, x_j) - \log \tfrac{\ell}{\varepsilon} \right)
     \qquad \text{for all nonempty } J \subseteq [\ell].
  \]
  Let $\td$ be an $\ell$-tuple of integers. 
  The set of elements $\tx$ in $S_\tm$ for which $\Delta(\varepsilon,m_j,x_j)= d_j$ 
  for all $j \in [\ell]$ has $(\tfrac{\varepsilon}{\ell}2^{m_1-d_1}, \dots, \tfrac{\varepsilon}{\ell}2^{m_\ell-d_\ell})$-small slices.
  
  Given the set $S_\tm$ as a list, we need to specify an output for a decompressor $\ttd'$ for some inputs $\varepsilon$ and $\ty$.
  We modify the tree-partitioning in the online algorithm of Theorem~\ref{th:main_componentwise_invertible} as follows.
  For each depth $j \in [\ell]$, for each node at depth $j$, for each $d_j \le m_j$, we create $2^{m_j - d_j}$ children, 
  and associate each string $w$ for which $\Delta(\varepsilon,m_j, w) = d_j$ to a random such child.

  We now consider the update procedure.
  The first step, the top-down percolation step, does not change. 
  In the bubble up step, we only change the third instruction. 
  For this instruction, we have no inverse function available. But instead, 
  we construct a decompressor $\tte$, and apply the decompressor $\tte'$ 
  obtained from the definition of $\Delta$-optimality. 
  This decompressor $\tte$ assigns programs of length $m_j - d_j$ to all $j$-th coordinates of tuples of children 
  with associated value~$d_j$. (Recall that there are precisely $2^{m_j-d_j}$ such children.)
  Let $\tte'$ be the corresponding decompressor obtained from  Definition~\ref{def:optimal_compressor}.
  If $\tte'(y_j) = x_j$ then $\tx$ is copied from the child to the parent.

  The analysis of this algorithm follows the analysis above: by the same reasoning as in  Lemma~\ref{lem:random_tree}, 
  for each fixed tuple $\td$, there are no collisions in leafs of $\td$-branches, 
  i.e., branches whose nodes at depth~$j$ have corresponding values for~$d_j$. 
  Indeed, each factor $\tfrac{\varepsilon}{\ell}2^{m_j-d_j}$ 
  in the condition of $(\tfrac{\varepsilon}{\ell}2^{\tm-\td})$-small slices, is a factor $\tfrac{\varepsilon}{\ell}$
  smaller than the number of branches, which is $2^{m_j-d_j}$.
  Finally, the probability that for a fixed $\tx \in S_\tm$, 
  on input $\varepsilon, \ttc_{\varepsilon,m_1}(x_1), \dots, \ttc_{\varepsilon, m_\ell}(x_\ell)$, 
  the tuple $\tx$ bubbles up to the root, also follows the same analysis. 
\end{proof}

\begin{remark}\label{remark:effective_decompression_distributed}
  Recall that in the single source case,
  the optimal compressor in the proof of Theorem~\ref{th:main_single_source}, 
  provides a modified decompressor that can be evaluated in space polynomial in $n_S = \max\{|z|: z \mathop{\in} S\}$, 
  when given oracle access to the original decompressor.
  Given oracle access to a decompressor $\ttd$ acting on tuples, 
  the above procedure defines a decompressor $\ttd'$ that can be evaluated in 
  time exponential in $\ell$ and $n_S = \max \{|z_j| : \tz \mathop{\in} S \text{ and } j \in [\ell]\}$, 
   and space polynomial in~$n_S \cdot 2^\ell$.
  This implies that if $\ttd$ is partial computable, respectively, computable and computable in exponential time, 
  then so is the corresponding~$\ttd'$. 
\end{remark}

\section{Lower bounds}\label{s:lowerbounds}

In this section we prove the lower bounds on the overhead $\Delta$ and on the number of random bits $r$ used by the compressors given in Proposition~\ref{prop:lowerbounds}.
To strengthen these lower bounds, we first show that $\Delta$-optimality is equivalent to a weaker notion of optimality.

\subsection*{Robustness of the definition of optimal compressor}

\textbf{Probabilistic decompressors.} 
We show that the definition of $\Delta$-optimality does not change if we also consider probabilistic decompressors, except for a small rescaling of~$\varepsilon$.

For a probabilistic partial function $\ttd$ let $\ttd_\maj$ be the partial function that for each argument of 
$\ttd$ is undefined if every value appears with probability at most $1/2$, and otherwise, it equals the unique 
value that appears with probability strictly more than $1/2$.

\begin{lemma}\label{lem:equivalence_optimality_probabilisticDecompressor}
  For every probabilistic function $\ff$:
  \[
  \Pr [\ttd(\ff(x)) \not= z] \;\ge\; \tfrac{1}{2} \cdot \Pr [\ttd_\maj(\ff(x)) \not= z].
  \]
\end{lemma}

\begin{proof}
 We call a value $p = \ff(x)$  {\em bad} if $\ttd(p) = z$ has probability at most~$1/2$.
 The right-hand side is equal to the probability that $\ff(x)$ generates a bad value. 
The inequality follows directly.
\end{proof}

\bigbreak
\noindent
\textbf{A single optimal decompressor.}
Recall that a Turing machine $\ttu$ is {\em optimal} if for every other Turing machine~$M$ 
there exists a constant $c_\ttd$ such that for all strings $x$: $\C_\ttu(x) \le \C_\ttd(x) + c_\ttd$.
We show that if in stead of considering any decompressor in the definition of $\Delta$-optimality, 
we only consider a single optimal decompressor, the definition does not change after a constant shift of~$\Delta$, 
provided~$\Delta$ is a computable expression.

\begin{lemma}\label{lem:equivalence_optimality_singleTM}
  Assume that $\Delta$ is a computable function of the compressors inputs, $\ttu$ is an optimal Turing machine, 
  and there exists a Turing machine $\ttu'$ such that for all $\varepsilon, x$ and $m \ge C_\ttu(x) + \Delta$ we have
  \[
  \Pr \left[ \ttu'(\ttc_{\varepsilon,m}(x)) = x \right] \ge 1-\varepsilon.
  \]
  Then $\ttc$ is $(\Delta + O(1))$-optimal.
\end{lemma}

\begin{proof}
 We show the contrapositive: If for large $c$, compressor $\ttc$ is not $(\Delta+2c)$-optimal, 
 then no machine $\ttu'$ satisfies the condition of the lemma. 
 
 For each positive integer $c$, 
 consider a decompressor $\ttd_c$ for which the $(\Delta + 2c)$-optimality condition does not hold, 
 i.e., there exists no $\ttd'$ that satisfies the condition of Definition~\ref{def:optimal_compressor}.
 By a compactness argument, there exists such a $\ttd_c$ with finite domain, and 
 hence, on input $c$, such a decompressor can be found by exhaustive search. 
 Note that this search requires the evaluating of~$\Delta$.

 We construct a decompressor $\ttd$ by combining all these $\ttd_c$. More precisely, 
 given an input of the form $0^c1p$, the output of $\ttd$ is given by $\ttd_c(p)$.
 By optimality of $\ttu$, we have 
 \[
 C_\ttu(x) \;\le\; C_\ttd(x) + c_\ttd \;\le\; C_{\ttd_c}(x) + c_\ttd + c. 
  \]
 Let $c = c_\ttd$. 
 The assumption $m \ge C_\ttu(x) + \Delta$ in the lemma implies the assumption $m \ge C_\ttd(x) + (\Delta + 2c_\ttd)$ 
 in the $(\Delta + 2c)$-optimality criterion for $\ttd_c$. Thus if there exists a machine $\ttu'$ 
 satisfying the conditions of the lemma, then it gives us a machine $\ttd'$ satisfying the $(\Delta+2c)$-optimality 
 criterion.  By choice of $\ttd_c$, such $\ttd'$ does not exist, and hence, $\ttu'$ also does not exist.
 The lemma is proven.  
\end{proof}

\subsection*{Lower bound for the randomness used by optimal compressors}

\begin{proposition*}[Second part of Proposition~\ref{prop:lowerbounds}]
  Let $\Delta$ and $r$ be functions of~$\varepsilon$ and~$n=|x|$.
  If there exists a $\Delta$-optimal compressor that can be evaluated with at most $r$ random bits in a computably bounded running time, then for all~$\varepsilon \le 1/2$:
  \[
    \left\lceil \varepsilon 2^{r+1} \right\rceil \;\ge\; \frac{n-r - \log (2/\varepsilon)}{\Delta + 4}
   \]
\end{proposition*}

\noindent
Recall that for a deterministic compressor we have $r=0$. 
If such a compressor is $\Delta$-optimal and has a computably bounded running time, 
then for $\varepsilon = 1/2$, this implies $\Delta \ge n-6$. 
As a warm-up, let us prove in a direct way, a slightly stronger bound for deterministic compressors. 

\begin{lemma*}
  If a $\Delta$-optimal compressor is deterministic, then $\Delta \ge n - 1$.
\end{lemma*}

\begin{proof}
  We show the lemma by contraposition.
  Assume there exists an $(n\mathop{-}2)$-optimal compressor.
  There are less than $2^n$ different outputs of length~$m = n-1$. 
  Thus given target $m$, there exist $2$ different $n$-bit $x$ and $x'$ that are mapped to the same string.
  Consider a decompressor~$\ttd$ 
  for which $\C_\ttd(x), \C_\ttd(x') \le 1$, and hence, $m \ge \C_\ttd(x) + (n-2)$. 
  By $(n\mathop{-}2)$-optimality, both $x$ and $x'$ must be decompressed correctly, 
  but this is impossible since they are mapped to the same output. 
  Hence, the compressor is not $\Delta$-optimal.
\end{proof}

\bigskip
\noindent
We now prove  the second item of  Proposition~\ref{prop:lowerbounds}.
Recall that $\Delta$-optimal compressors define conductors, and hence condensers, with small 
output size, see Lemma~\ref{lem:compressor_to_conductor}. 

In the pseudorandomness literature, condensers are usually defined as $2$-argument functions
$f \colon \mcX \times \mcD \rightarrow \mcY$, where the size of $\mcD$ is called the {\em degree}, 
see also appendix~\ref{s:standard_def}. For such $f$, the function $x \mapsto f(x,U_\mcD)$ is a probabilistic 
1-argument function as above. 
The number of random bits needed for its evaluation is $\log \# \mcD$, provided $\# \mcD$ is a power of 2.

We obtain lower bounds for the amount of randomness from degree lower bounds of condensers, 
such as given in~\cite{nis-zuc:j:extract,rad-tas:j:extractors}. 
In particular, the second item of Proposition~\ref{prop:lowerbounds} follows by applying the next 
lower bound for a suitable value of~$K$.

\newcommand{\propDegreeLowerboundCondensers}{
  If $f: \mcX \times [D] \rightarrow \mcY$ is a $(K,\varepsilon)$-condenser 
  with $\# \mcY \ge 2K \ge 4D/\varepsilon$ and $\varepsilon \le 1/2$, then
  \[ \wider
  \lceil 2\varepsilon D\rceil \ge \frac{\log ({\# \mcX}/{K})}{3 + \log (\# \mcY/K)}.
  \]
}

\begin{proposition}\label{prop:degreeLowerboundCondensers}
  \propDegreeLowerboundCondensers
\end{proposition}

\noindent
The proof is given in appendix~\ref{s:condenserLowerbounds}.

\begin{proof}[Proof of the second item of Proposition~\ref{prop:lowerbounds}.]
  If $\ttc$ is a $\Delta$-optimal compressor, then for all $n,m$ and $\varepsilon$ the function
  $
  \ff : \{0,1\}^n \rightarrow \{0,1\}^{m} 
  $
  defined by $\ff(x) = \ttc_{\varepsilon,m}(x)$ is a $(2^{m-\Delta},\varepsilon)$-conductor, and hence, a $(2^{m-\Delta},\varepsilon)$-condenser.
  We apply the lower bound of Proposition~\ref{prop:degreeLowerboundCondensers} 
  with $D = 2^r$ and $K = 2D/\varepsilon$, for the target size given by $m = \lceil\log K\rceil + \Delta$. 
  Thus, we have $\mcY = \{0,1\}^m$ and $\log (\#\mcY/K) \le \Delta + 1$. 
  The formula of Proposition~\ref{prop:degreeLowerboundCondensers} implies the inequality of the second item.
\end{proof}

\subsection*{Lower bound for the overhead of optimal compressors}

The proof of the first item of Theorem~\ref{prop:lowerbounds}, uses the second item and 
Lemma~\ref{lem:reduce_randomness_invertibleFunc}  below. This lemma essentially states that
a simple modification of any 
invertible function can reduce its use of randomness
to $r \le \Delta + O(1)$.  In the proof we use another lemma.

\begin{lemma}\label{lem:sampler}
  Let $b$ be a nonnegative integer.
  If a measure $\mu$ over $\mcY$ is supported on a set of size at most $2^d$, 
  then there exists a sampling algorithm that uses $d+b$ random bits, 
  and generates samples from $\mcY$ such that each element $y \in \mcY$ appears
  with probability at most $(1+3 \cdot 2^{-b})\mu(y)$. 
\end{lemma}

\begin{proof}
  Let $D = 2^d$.
  The proof for arbitrary $b$ is similar to the proof for $b=0$, which we present here.
  We show that there exists a measure $\nu$ whose values are multiples of $1/D$ 
  such that $\nu(y) \le 4\mu(y)$ for all $y \in \mcY$. This is done by rounding.
  The sampling procedure satisfying the conditions of the lemma, simply outputs a random sample of~$\nu$, 
  and we can generate such a sample using $d$ random bits.\footnote{ 
    The following sample procedure can be used: 
    pick a random integer~$i$ in the interval~$[2^d]$, 
    and output the largest~$y$ in the support of~$\nu$ with $\sum_{y' \prec y} \nu(y') \le i/2^d$, 
    where $\prec$ is some fixed order on~$\mcY$. 
  }

  Let $\mu'(y)$ be $0$ if $\mu(y) < 1/2D$, and $\mu(y)$ otherwise.
  Note that $Z = \sum_y \mu'(y) \ge 1/2$ since $\mu$ is supported on at most $D$ elements.
  $\mu'/Z$ is a measure. We obtain the measure $\nu$ by rounding all values of $\mu'/Z$ to multiples of $1/D$.
  By rounding up a suitable set of elements and rounding down the others, we indeed obtain a measure, because:
  \\- if we round up all values, the sum of $\sum_y \nu(y)$ is at least $1$, 
  \\- if we round down, the sum is at most $1$, and
  \\- changing the choice for one element $y$, changes the sum by a multiple of~$1/D$.
  \\We verify that $\nu(y) \le 4\mu(y)$.
  Recall that $1/Z \le 2$. If $\mu(y) \ge 1/(2D)$, then 
  \[
  \wider
  \nu(y) \le \frac{\mu'(y)}{Z} + \frac{1}{D} \le 2\cdot\mu(y) + 2 \cdot \mu(y).
  \qedhere
  \]
\end{proof}

\begin{lemma}\label{lem:reduce_randomness_invertibleFunc}
  Assume $\varepsilon \le 1/4$. 
  For every $(K,\varepsilon)$-invertible function with at most $M$ values, 
  there exists a $(K,2\varepsilon)$-invertible function with at most $M+K$ 
  values that can be evaluated using at most $\lceil \log (M/K) \rceil + 3$ bits of randomness.
\end{lemma}

\begin{proof}
  Let $\ff$ be $(K,\varepsilon)$-invertible with inverse~$g$. 
  For a given input~$x$, consider the list of all possible outputs that appear with positive probability.
  Let $P_x$ be the set of these outputs after removing a set of measure at most~$\varepsilon$ with maximal cardinality.
  We first prove the following 2 properties.
  \begin{itemize}[leftmargin=*]
    \item 
      Let $S$ be a list of $K$ inputs.
      For all $x \in S$: $\# P_x$ is at least the number of values $y$ for which~$g_S(y) = x$.

    \item 
    There are less than $K$ inputs $x$ for which $\# P_x > M/K$. 
  \end{itemize}
  For the first property, note that by definition of $(K,\varepsilon)$-inverse, for a given $x$, 
  the measure of outcomes $y = \ff(x)$ with~$g_S(y) \not= x$ is at most~$\varepsilon$.
  Thus, the property follows by choice of~$P_x$.

  For the second property, consider the set of inputs $x$ for which the property holds. 
  Assume there are at least $K$ such inputs, and let $S$ be a list of precisely $K$ 
  such inputs. By the first property, for each such element there are 
  more than $M/K$ outputs $y$ with $g_S(y)=x$. 
  But since there are only $M$ outputs and $g_S$ is deterministic, this is impossible. 
  Our assumption must be false, and hence there are less than $K$ such inputs. 
  The 2 properties are proven.

  \medskip
  \noindent
  {\em Construction of the $(K,2\varepsilon)$-invertible function~$\ff'$.}
  \\Let $y_1, \dots, y_K$ be new values, different from any value of~$\ff$.
  On input $x$ consider the set $P_x$.
  If $\# P_x > M/K$, compute the index $K'$ of $x$ in a list of all such inputs,
  and let the output be $y_{K'}$. (The second property implies $K' \le K$.)
  Otherwise, let $\mu_x$ be the measure obtained by normalizing the probabilities of $P_x$. 
  Sample an output $y$ from $\ff(x)$ using a sample procedure 
  satisfying the conditions of Lemma~\ref{lem:sampler} with $b=3$.

  \medskip
  \noindent
  By construction, this algorithm uses at most $\lceil \log (M/K) \rceil + 3$ bits of randomness.

  \medskip
  \noindent
  {\em Construction of the inverse~$g'$ of~$\ff'$.} 
  \\For values $y = y_{K'}$, the value of $g'_S(y)$ is the $K'$-th 
  input for which $\# P_x > M/K$, (thus in this case, the value does not depend on $S$).
  For the other values of $y$, we simply output~$g_S(y)$.

  \medskip
  \noindent
  It remains to bound the probability that $g'_S(\ff'(x)) \not= x$ by~$2\varepsilon$.
  For $x$ such that $\# P_x > M/K$, this probability is~$0$, by construction.
  For the other $x$, the rescaling in the definition of $\mu_x$ amplifies the 
  probability of $g_S(\ff(x))$ by a factor at most~$1/(1-\varepsilon) \le 4/3$.
  The sampling procedure from Lemma~\ref{lem:sampler} with $b=3$ can amplify this probability 
  by at most a factor~$1+3/8$. Together, the amplification is at most a factor~$2$.
  The invertibility condition is satisfied.
\end{proof}

\begin{proposition*}[First part of Proposition~\ref{prop:lowerbounds}, restated]
      $
      \Delta \ge \log \tfrac{n}{\epsilon} - \log \log \tfrac{n}{\epsilon} - 8 
      $
      \,if\, $2^{-n/4} \le \epsilon \le 1/4$.
\end{proposition*}

\begin{proof}
  For a fixed target size~$m$, the compressor provides us with a $(2^{m-\Delta},\epsilon)$-invertible function.
  From Lemma~\ref{lem:reduce_randomness_invertibleFunc} we obtain a $(2^{m-\Delta}, 2\epsilon)$-invertible function 
  that can be evaluated with~$r = \Delta+3$ random bits, and whose output is 1 bit longer.
  To this function we apply the same argument as in the proof of the second item of the Proposition.
  We obtain that
  \[
    \left\lceil \epsilon 2^{\Delta+5} \right\rceil \;\ge\; \frac{n-\Delta-4 - \log (1/\epsilon)}{\Delta + 4}.
   \]
   Note that the right-hand side is decreasing in $\Delta$ and the left-hand side is increasing.
  Hence, any value of $\Delta$ that makes the above equation false is a lower bound for the overhead of a decompressor. 
  A calculation shows that for all values of $n$ and $\epsilon$ with $2^{-n/4} \le \epsilon \le 1/4$,
  the assignment
  \[
  \Delta \;=\; \log \tfrac{n}{\epsilon} - \log \log \tfrac{n}{\epsilon} - 8.
  \]
  violates the inequality. This finishes the proof.
\end{proof}

\subsection*{Limitations for polynomial time compression}

\begin{question*}
  Suppose there exists an $O(\log \tfrac{n}{\varepsilon})$-optimal polynomial time computable compressor, 
  does this give us an explicit conductor of logarithmic degree?
\end{question*}

\noindent
We do not know the answer. 
But, we can prove something weaker: such a compressor 
implies the existence of a family of conductors 
of degree $O(2^\Delta)$ that can be computed by polynomial size circuits. 

\smallskip
We sketch the proof. 
We need to prove a variant of Lemma~\ref{lem:reduce_randomness_invertibleFunc} 
where the first invertible function is computable in polynomial time, and the second 
computable by a family of polynomial sized circuits.

Imagine, we want to evaluate the function $\ff'$ constructed in the above proof in polynomial time. 
To do this, we face 2 obstacles. 
The first is that counting the number of $n$-bit input strings $x'$ that lexicographically precede $x$, 
and satisfy $\# P_x > A$ for the given bound~$A$, takes exponential time. To solve this, 
we consider strings $x$ with $\# P_x > \mathbf{2}A$, and observe that 
at most $K/2$ strings satisfy this condition.
We handle such strings by recursively 
applying the algorithm to a $(K/2,\varepsilon)$-invertible function.

A second obstacle is that one needs to decide whether $P_x$ has size larger than~$A$, and 
if the algorithm for $\ff$ uses a very large amount of randomness, we can not brute-force 
search all random choices in polynomial time.
However, if $\Delta$ is logarithmic, then $A$ is polynomial. 
If $\# P_x$ is not too much above $A$, we can approximate $\# P_x$ in polynomial time
with a probabilistic algorithm. For this, we try a few random seeds, and inspect the frequencies 
with which the output values are sampled. With minor modification, such an approximation 
is enough to execute the plan above.

Finally, with a standard technique, we transform probabilistic algorithms to circuits. 
We obtain circuits of size $\poly(\tfrac{n}{\varepsilon}, 2^{\Delta})$ 
which map $n$-bit strings to $m$-bit strings, and 
compute $(2^{m-\Delta},\varepsilon)$-conductors. 

\section{Acknowledgements}
The first author is grateful to Andrei Romashchenko and Sasha Shen for useful discussions and for their insightful suggestions.

\bibliography{theory-3}
\bibliographystyle{alpha}
\appendix

\section{Proof of Lemma~\ref{lem:simple_inverse}}\label{sec:proof_simple_inverse}

\begin{lemma*}[Restated]
  \lemmaSimpleInverse
\end{lemma*}

\begin{proof}
  Note that by assumption on the size of $\mcY$, there exists at least one $x \in \mcX$ such that 
  \[
  F_x \; \subseteq \; \bigcup_{\bm{z \in \mcX}, z \ne x} F_z\;, 
  \]
  because if this was not the case, 
  we could build a 1-to-1 correspondence between elements of $\mcX$ and some elements of $\mcY$, 
  by mapping $z \in \mcX$ to an element in $F_z$ that appears in no other set $F_{z'}$,
  and this contradicts $\# \mcY < \# \mcX-K/2$. 
  In fact, by the same reasoning, there exist $K/2$ such elements $x_1, \dots, x_{K/2}$. 

  If there exists such an $x_i$ for which $F_{x_i}$ has size less than $K$, then 
  we choose $x = x_i$ and obtain $S$ by selecting for each $y \in F_x$ an element $z \in \mcX$ with $y \in F_z$.
  Otherwise, if each such set has size at least~$K$,  
  we construct $S \subseteq \mcX$ together with a set $Y \subseteq \mcY$ iteratively as follows.
  Initially, $S$ and $Y$ are empty. At stage $i$, we check whether adding $F_{x_i}$ to $Y$ 
  increases its cardinality by at least~$K/2$. If this is indeed so, we add $x_i$ to $S$ and $F_{x_i}$ to $Y$, and proceed to stage $i+1$.
  Otherwise, we choose $x = x_i$, and exit the iterative process. 
  Finally, for each  $y \in F_{x} \setminus Y$, we add an element $z$ to~$S$ for which $y \in F_z$.

  In the first case, the inclusion of the lemma is satisfied by construction. 
  In the second case, we first need to verify
  that there exists an index $i \le K/2$ for which the cardinality of $Y$ increases by at most $K/2$.
  Indeed, assume that for all such $i$ the cardinality increases by more than this amount, then after $K/2$
  stages, the cardinality of  $Y$ exceeds $(K/2) \cdot (K/2)$. On the other hand, $Y \subseteq \mcY$ and 
  by assumption, the cardinality of $Y$ is at most $K^2/4$. A contradiction. 
  This implies that $S$ contains less than $K/2$ elements of the form~$x_i$. 
  In the last step, we add at most $\# (Y \setminus F_{x})$ elements to $S$, and by construction this is at most $K/2$.
  In total $S$ contains less than $K$ elements. 
  Also, by construction, the inclusion of the lemma is satisfied.
\end{proof}

\section{Standard definitions of conductors and condensers}\label{s:standard_def}

We present the definitions of condensers and conductors from the literature,
and explain their equivalence with our versions.

Let $X$ be a random variable in $\mcX$ with probability measure~$P$.
The {\em minentropy} is 
\[
H_{\infty}(X) \;=\; \min \{\log (1/P(x)) : x \in \mcX\}. 
\]
We use the statistical distance to measure similarity of random variables.
Let $Y$ and $Y'$ be random variables having the same range $\mcY$. These variables are said 
to be $\varepsilon$-close if $\left|\Pr [Y \mathop{\in} T] - \Pr [Y' \mathop{\in} T] \right| \;\le\; \varepsilon$ for all $T \subseteq \mcY$.

\begin{remark*}
  In this definition, the maximal difference is obtained for the set $T$ given by all values $y \in \mcY$ for which 
  the probability of~$Y$ exceeds the probability of~$Y'$.
  Hence, if $Y$ and $Y'$ have statistical distance at most $\varepsilon$, then for any~$\gamma$, 
  the $\gamma$-excess of $Y$ and $Y'$ differ by at most~$\varepsilon$.
\end{remark*}

\noindent
A random variable with minentropy at least $k$ is called a  {\em $k$-source}. A random variable that 
is $\varepsilon$-close to a $k$-source, is called a {\em $(k,\varepsilon)$-source.}

\begin{lemma}\label{lem:equivalence_excess}
  Let $Y$ be a random variable in a domain of size at least $K$. 
  $Y$ is a $(\log K,\varepsilon)$-source if and only if it has $(1/K)$-excess at most~$\varepsilon$. 
\end{lemma}

\begin{proof}
  Assume that $Y$ is a $(\log K,\varepsilon)$-source. 
  By definition, there exists a $(\log K)$-source $Y'$ with statistical distance at most~$\varepsilon$ from~$Y$.
  Since the $(1/K)$-excess of $Y'$ is zero, the remark above implies that
  $Y$  has $(1/K)$-excess at most~$\varepsilon$. 
  Note that this direction does not use the assumption on the domain size.

  Now assume that $Y$ has $(1/K)$-excess at most $\varepsilon$. 
  Let $Y'$ be a $(\log K)$-source obtained by trimming the measure of $Y$ to the value $1/K$, 
  and redistributing all trimmed measure over the other values, while keeping them bounded by~$1/K$. 
  This is possible, because the domain of $Y$ is at least~$K$. 
  The variables $Y$ and $Y'$ are $\varepsilon$-close, because for the 
  optimal set $T$ of the remark, the difference is precisely the $(1/K)$-excess of~$Y$.
  Hence, $Y$ is a $(\log K,\varepsilon)$-source.
\end{proof}

\begin{lemma}\label{lem:equivalence_condensers}
  A function $x \mapsto f(x,U_{\mcD})$ is a \condenser{K}{K'} if and only if for every $(\log K)$-source~$X$, the variable $f(X,U_{\mcD})$ is a $(\log K',\varepsilon)$-source. 
\end{lemma}

\begin{proof}
  The backwards implication follows directly from  Lemma~\ref{lem:equivalence_condensers}: for every set $S$ of size $K$, 
  the function $f(U_S, U_\mcD)$ is a $(K',\varepsilon)$-source, and hence, has $(1/K')$-excess at most~$\varepsilon$.

  For the forward direction, observe that each random variable that is a $(\log K)$-source, can be written as a mixture 
  of uniform distributions on sets of size~$K$.\footnote{
    A nice proof from~\cite[Lemma 6.10]{vad:b:pseudorand} is as follows. Let $P_1, P_2, \dots, P_e$ be probabilities on a finite set such that all $P_i \le 1/K$. 
    Construct a circle with circumference 1, and lay out consecutive segments $I_1, I_2, \dots, P_e$ on its circumference of lengths, respectively,  $P_1, P_2, \dots$.  
    Consider a $K$-regular polygon  inscribed in this circle.  The $K$ vertices of the polygon  lie on different segments and therefore the polygon specifies $K$ segments. 
    For a  subset $S$ of $K$ segments, let $\alpha_S$ be the probability that a random rotation of some fixed polygon specifies $S$.
    The result follows by showing that the distribution $P$ is the same as the distribution of $\sum \alpha_S U_S$, 
    where the sum is taken over all $K$-element sets of segments. Indeed, the probability that we obtain segment $I_t$ by first selecting a random rotation of the fixed polygon, followed by a random selection of one the $K$  segments specified by the rotated polygon is the same 
    as the probability of obtaining $I_t$ by selecting a random point on the circle and taking the segment containing the point, and the latter probability is equal to~$P_t$. 
  }
  Each such variable induces a $(\log K',\varepsilon)$-source. 
  Therefore, the mixture induces a mixture of $(\log K',\varepsilon)$-sources, which is also a $(\log K',\varepsilon)$-source.
\end{proof}

\begin{remark}\label{remark:extractor_vs_condenser}
  A probabilistic function $\ff \colon \mcX \rightarrow \mcY$ is a {\em $(K,\varepsilon)$-extractor} 
  if and only for every $(\log K)$-source, the variable $\ff(X)$ is $\varepsilon$-close to~$U_\mcY$.
  There is only 1 distribution over $\mcY$ that is a $(\log \# Y)$-source, which is the uniform distribution.
  Hence, $\ff$ is a $(K,\varepsilon)$-extractor if and only if it is a~\condenser{K}{\# \mcY}.
\end{remark}

\begin{definition}[\cite{cap-rei-vad-wig:c:conductors}]\label{def:simpleConductor}
  A function $f \colon \{0,1\}^n \times \{0,1\}^d \rightarrow \{0,1\}^m$ is a $(k_\textnormal{max},\varepsilon,a)$-{\em simple conductor} 
  if for any nonnegative integer $k \le k_\textnormal{max}$ and any $k$-source $X$, the induced variable $f(X, U_d)$ is a $(k+a,\varepsilon)$-source.
\end{definition}

\begin{lemma}\label{lem:conductorLiterature_implies_ours}
  If $f$ is a $(\kmax,\varepsilon,1)$-simple conductor, 
  then $x \mapsto f(x,U_d)$ is a $(2^{\kmax},\varepsilon)$-conductor according to  Definition~\ref{def:conductor}.
\end{lemma}

\noindent
Note that it is not enough to assume that $f$ is a $(\kmax, \varepsilon, 0)$-simple conductor, because the definition only 
considers integers $k$, and for sets $S$ whose size is not a power of $2$, the output might loose 1 bit of minentropy. 
The proof follows directly from Lemma~\ref{lem:equivalence_condensers}.


To a probabilistic function $\ff \colon \{0,1\}^n \rightarrow \{0,1\}^m$ whose evaluation requires $r$ random bits, 
we associate the function 
\[
f \colon \{0,1\}^n \times \{0,1\}^r \rightarrow \{0,1\}^m
\]
obtained by setting the random bits equal to the second argument. 

\begin{lemma}\label{lem:ours_implies_conductorLiterature}
  Let $\kmax \le m$. 
  If $\ff\colon \{0,1\}^n \rightarrow \{0,1\}^m$ is a $(2^{\kmax},\varepsilon)$-conductor according to  Definition~\ref{def:conductor},
  then the associated function $f$ is a $(\kmax,\varepsilon,0)$-simple conductor. 
\end{lemma}

\noindent
Again this follows directly from the definitions and Lemma~\ref{lem:equivalence_condensers}.

\section{Non-explicit conductors} \label{s:randcondenser}

Recall that $[n] = \{1,2,\dots,n\}$. Proposition~\ref{prop:existsPrefixExtractor} follows from the following.

\begin{proposition*}
  For all $\varepsilon$, $n$ and $\Kmax$, there exists a $(\Kmax,\varepsilon)$-conductor 
  $\ff \colon [2^n] \rightarrow [4\Kmax]$ whose evaluation requires at most $\log \tfrac {4n} \varepsilon$ bits of randomness. 
\end{proposition*}

\noindent
 In the proof  we use a theorem by Hoeffding. The following is obtained from equation 2.1 of Theorem 1 in~\cite{hoe:j:probbound}.

\begin{theorem*}[Hoeffding bound, simplified]\label{th:Hoeffding}
  Let $X_1, \dots, X_t$ be independent variables taking values in the real interval~$[0,1]$. Let $\mu$ be 
  the expected value of $Z = X_1 + \dots + X_t$, for all $\nu > \mu$:
  \[
  \Pr \left[ Z \ge \nu \right] \;\le\; \left(\frac{\mu}{\nu}\right)^\nu.
  \]
\end{theorem*}

\begin{proof}
  Let $\mcY = [4\Kmax]$. 
  If $\# \mcY > 2^n$, the identity function $\ff$ satisfies the conditions. If $\Kmax \le 1/\varepsilon$, then we output a random element of~$[2^{\lceil \log (1/\varepsilon) \rceil}] \subseteq \mcY$.  Assume $\varepsilon \Kmax > 1$ and $\#\mcY \le 2^n$.
  
  We construct a function $\ff$ such that 
  \[
  \forall X \subseteq [2^n] \text{ with } \# X \le \Kmax \;\; \forall Y \mathop{\subseteq} \mcY \text{ with } \# Y = \lceil \varepsilon \# X\rceil
  \;\; : \;\; \Pr \left[ \ff(U_X) \mathop{\in} Y \right] \;\le\; \varepsilon.
  \]
  We first show that such an~$\ff$ is a conductor. Indeed, for any $K \le \Kmax$ and any set $X \subseteq [2^n]$ of size~$K$, consider 
  the set~$Y$ of elements $y \in \mcY$ for which $y = \ff(U_X)$ has probability strictly more than~$1/K$.
  We have~$\# Y \le \varepsilon K$, because otherwise any subset~$Y' \subseteq Y$ of size~$\lceil \varepsilon K \rceil$ 
  violates the condition above. 
  This implies that the $(1/K)$-excess of $\ff(U_X)$ is at most~$\varepsilon$.

  \smallskip
  Let $D = \lceil 3n/ \varepsilon \rceil$.
  We obtain $\ff$ from a function $f \colon [2^n] \times [D] \rightarrow \mcY$
  as $\ff(x) = f(x,U_{[D]})$. The existence of the required $f$ is shown by the probabilistic method:
  we select each value $f(x,i)$ randomly in $\mcY$, and show that the property is satisfied with positive probability.

  For a fixed $(x,i)$ and a fixed set $Y$ of size at most $\varepsilon \Kmax$, such a random $f$ satisfies $f(x,i) \in Y$ with probability at most~$\varepsilon/2$, 
  since $\Kmax \varepsilon > 1$ and hence, $\# Y \le \lceil \varepsilon \Kmax\rceil \le 2\varepsilon \Kmax \le \varepsilon \# \mcY/2$.
  For fixed sets $X$ and $Y$ satisfying the conditions, consider the event 
  ``{\em $f(x,i) \in Y$ for more than an $\varepsilon$-fraction of $(x,i) \in X \times [D]$}''.
  By the simplified Hoeffding bound, this  is at most $2^{-\varepsilon D \# X}$. 

  The number of sets $X$ of size $K$ is at most $2^{nK}$. 
  The number of sets $Y$ of size $\lceil \varepsilon K\rceil$ is $(\# \mcY)^{\lceil\varepsilon K\rceil} \le 2^{n\lceil \varepsilon \rceil K}$. 
  By the union bound, the probability that the condition is violated is at most:
  \[
  \le \; \sum_{K \ge 1}^{\Kmax} 2^{nK} \cdot 2^{nK} \cdot 2^{-\varepsilon KD}\;. 
  \]
  By choice of $D$, each term is at most $2^{-nK}$, 
  thus the sum is strictly smaller than $1$. This implies that the required function $f$ exists.
  %
\end{proof}

\section{Construction of explicit conductors}\label{s:conductor_construction}

We prove  Proposition~\ref{prop:conductors1}. 
For inductive purposes, it is convenient to switch to a stronger type of conductors, 
which correspond to ``conductors of the extracting type'' in Capalbo et al.~\cite{cap-rei-vad-wig:c:conductors}. 

\begin{definition}\label{def:extracting_conductors}
 A probabilistic function $\ff \colon \mcX \rightarrow \mcY$ is an $\varepsilon$-conductor if it is an $(\# \mcY,\varepsilon)$-conductor.
\end{definition}

\noindent
We restate the proposition using this definition.

\begin{proposition*} 
  There exists an explicit family $\ff^{\smallerbrackets{n}}_{\varepsilon, k}$ 
  of $\varepsilon$-conductors $\ff^{\smallerbrackets{n}}_{\varepsilon, k}\colon \{0,1\}^n \rightarrow \{0,1\}^k$ whose 
  evaluation requires $O(\log k \cdot \log \tfrac{n}{\varepsilon})$ random bits, for all $\varepsilon$, $k$ and~$n$. 
\end{proposition*}

\noindent
Every $\varepsilon$-conductor is a \condenser{\# \mcY}{\# \mcY},
and hence a $(\mcY,\varepsilon)$-extractor, by Remark~\ref{remark:extractor_vs_condenser} in appendix~\ref{s:standard_def}.
Raz, Reingold and Vadhan~\cite{rareva:c:extractor}, and Guruswami, Umans, and Vadhan~\cite{guv:j:extractor} construct explicit extractors 
from which we immediately obtain condensers. Both constructions actually give conductors.
The latter uses the smallest amount of randomness for all choices of the parameters~$n,k,\varepsilon$. 
Therefore, we explain why this construction also provides conductors.

To prove the main result in~\cite{guv:j:extractor}, several extractors from Theorem~\ref{th:GUV_partial} (restated below)
are concatenated. 
To show that such a concatenation provides 
larger extractors, Guruswami et al.  use a known ``composition lemma", see~\cite[Lemma 4.18]{guv:j:extractor}. 
We show that this concatenation also provides conductors, and for our purposes, an easier variant of the composition lemma is enough. 
The above proposition follows almost immediately from the next 2 results.

\begin{theorem*}[\cite{guv:j:extractor},  Theorem 1.5 or 4.17]
  \lemmaGUVpartial
\end{theorem*}

\begin{lemma}[Composition lemma]\label{lem:compositionConductors}
  If $\fs \colon \mcX \rightarrow \mcY_1$ 
  is a \condensser{K'}{\# \mcY_1}
  and $\ft\colon \mcX \rightarrow \mcY_2$ is an $\varepsilon_2$-conduc\textbf{t}or with $\# \mcY_2 \ge K'$, 
  then 
  \[
  x \;\;\longmapsto\;\; \left(\fs(x), \ft(x)\right)
  \] 
  is an $(\varepsilon_1 {+} \varepsilon_2)$-conductor.
\end{lemma}

\noindent
We first prove the proposition, and afterwards this lemma. 

\begin{proof}[Proof of Proposition~\ref{prop:conductors1}.]
  For $k \ge n$ the function $\ff(x) = x$ satisfies the conditions. For $k \le 2$, the construction of an $\varepsilon$-conductor 
  that uses $O(1)$ of randomness is also easy.
  Otherwise, we repeatedly apply  Lemma~\ref{lem:compositionConductors} to the condenser of  Theorem~\ref{th:GUV_partial}, where we choose ${K'} = \#\mcY_2$,
  as long as the new output size is at most~$k$.
  In each application, the bit length of the output increases at least by a factor~$3/2$.
  Hence, $O(\log k)$ applications are sufficient.
  After this, we do one more application, where ${K'} < \# \mcY_2$ is chosen to obtain the correct output size.
  We obtain an $O(\varepsilon \log k)$-conductor that uses $O(\log k \cdot \log \tfrac{n}{\varepsilon})$ bits of randomness.
  The proposition follows after downscaling~$\varepsilon$ by a factor~$O(\log k)$.
\end{proof}

\noindent
To prove the lemma, we use a direct corollary of Lemma~\ref{lem:equivalence_condensers}.

\begin{corollary*}\label{cor:equivalence_condenser}
  If $\ff \colon \mcX \rightarrow \mcY$ is a {\em \condenser{K'}{M}}, then for every set $X \subseteq [\mcX]$ of size at least~$K'$, 
  the variable $\ff(U_X)$ has $(1/M)$-excess at most~$\varepsilon$. 
\end{corollary*}

\begin{proof}[Proof of Lemma~\ref{lem:compositionConductors}.]
  Let $X \subseteq \mcX$ be a set of size at most~$\#\mcY_1 \cdot \#\mcY_2$.
  We need to prove that the $(1/\# X)$-excess of $(\fs(U_X), \ft(U_X))$ is at most $\varepsilon_1 + \varepsilon_2$.
  If $\#X \le K'$, then $\# X \le \#\mcY_2$, and already the second component has $(1/K)$-excess at most $\varepsilon_2$. 
  The first component can only decrease the excess, and hence the property is satisfied.
  Assume $\# X \ge K'$. By the condenser property and the corollary above, the variable~$\fs(U_X)$ has $(1/\# \mcY_1)$-excess at most~$\varepsilon_1$.
  For a list $L$, let $U_L$ represent a randomly selected element from~$L$. 
  Our plan is to bound the excess of the pair using the following observation.

  \begin{quote}
    \textit{Let $L$ be a list and $b$ an integer. Assume $L'$ is a maximal sublist of $L$ that contains the same element at most $b$ times. 
    Thus, if an element $y$ appears at most $b$~times in~$L$, the list~$L'$ contains all appearances, and otherwise precisely~$b$ appearances.
    Then, the $(b/\# L)$-excess of~$U_L$ is equal to the probability~$\Pr [U_L \not\in L']$.
  }
  \end{quote}

  \noindent
  We assume that $\fs(x)$ can be evaluated by a deterministic function $S \colon \mcX \times \mcD_1 \rightarrow \mcY_1$ 
  by setting $\fs(x) = S(x, U_{\mcD_1})$. Moreover, we assume that $\mcD_1$ is finite. Let~$D_1 = \# \mcD_1$. 
  If this is not the case, we use an approximation for a finite $\mcD_1$ and take limits at the end of the proof.
  Similarly, assume $\ft$ is evaluated by $T \colon \mcX \times \mcD_2 \rightarrow \mcY_2$ with finite $D_2 = \# \mcD_2$.

  Consider the list $L$ containing all pairs $(S(x,i_1), T(x,i_2))$ for all $x \in X$, $i_1 \in \mcD_1$ and~$i_2 \in \mcD_2$, 
  which has length~$D_1D_2\# X$.
  We need to show that $U_L$ has $(1/\# X)$-excess at most $\varepsilon_1 + \varepsilon_2$. 
  We construct a sublist $L'$ that satisfies the conditions of the observation above.

  \begin{itemize}[leftmargin=2em]
    \item Let $L_1$ be the list of pairs $(S(x,i_1), x)$ for all $x \in X$ and~$i_1 \in \mcD_1$.

    \item For each $y_1 \in \mcY_1$, let $L_{y_1}$ be the list of the first $D_1\# \mcY_2$ pairs of the form $(y_1,\cdot)$ in $L_1$. 

    \item Split each $L_{y_1}$ in $D_1$ sublists of size at most $\# \mcY_2$ such that in each sublist, each right coordinate is unique.
      Denote these lists as $L_{y_1, e}$ for $e = 1, \dots, D_1$.

    \item Let $P_{y_1,e}$ be the list of pairs $(y_1, T(x,i_2))$ for all $(y_1, x) \in L_{y_1,e}$ and~$i_2 \in \mcD_2$.

    \item Let $P'_{y_1,e}$ be the sublist containing all $D_2$ first appearances of pairs of the form~$(\cdot, y_2)$.

    \item Let $L'$ be the concatenation of all lists $P'_{y_1,e}$ for $y_1 \in \mcY_1$ and~$e \le D_1$.
  \end{itemize}
  By construction, each pair $(y_1, y_2)$ appears at most $D_1D_2$ times in $L'$. 
  It remains to show that $\Pr[U_L \not\in L'] \le \varepsilon_1+\varepsilon_2$.
  Note that if in the second and fifth step, we did not restrict the number of appearances, then we would obtain $L = L'$.
  We show that in these steps, we withhold at most an $\varepsilon_1$-fraction and $\varepsilon_2$-fraction of the elements in~$L$, 
  and this finishes the proof.

  For the second step, note that since $\# \mcY_2 \ge \# X/ \#\mcY_1$ and $\# L_1 = D_1\# X$, 
  we have $D_1 \# \mcY_2/\# L_1 \ge 1/\# \mcY_1$, and hence, 
  the $(D_1 \# \mcY_2/\# L_1)$-excess of~$U_{L_1}$ is at most~$\varepsilon_1$. 
  By the observation, at most an $\varepsilon_1$-fraction of elements are withhold.

  For the fifth step, note that the right coordinates in each~$L_{y_1,e}$ define a set~$X'$ of size at most~$\# \mcY_2$. 
  By the conductor property, the variable $\ft(U_{X'})$ has $(1/\# X')$-excess at most~$\varepsilon_2$. 
  By the observation, each list $P'_{y_1, e}$ contains 
  all but an $\varepsilon_2$-fraction of the elements of $P_{y_1, e}$. 
  Hence, this step withholds at most an $\varepsilon_2$-fraction of the elements.
\end{proof}

\section{Proof of  Proposition~\ref{prop:invertible_two_functions}}\label{s:sw-easy}


\begin{proposition*}[Restated]
  \propInvertibleTwoFunctions{}
\end{proposition*}

\noindent
Let the sets in the second and third condition be denoted as $S_{x_2}$ and~$S_{x_1}$. Let 
$g_1(S_1, \cdot)$ denote an inverse of $\ff_1$ in $S_1 \subseteq \mcX_1$ and similar for~$g_2$.
We need to determine an inverse $g$ of the product function. More precisely, 
for every $(y_1, y_2) \in \mcY_1 \times \mcY_2$, we need to specify a value $g(S, y_1, y_2)$ that satisfies 
the conditions of $(3\varepsilon)$-invertibility.

A first idea is to define the set $T$ containing all $(x_1,x_2)$ for which $x_1 = g_1(S_{x_2}, y_1)$.
Thus $T$ contains at most one pair of the form $(\cdot, x_2)$ for all $x_2$.
Let $T_2$ be the projection of $T$ on the second coordinate. 
We apply $g_2(T_2, y_2)$, and if the outcome is $x_2$, we output the corresponding pair~$(x_1, x_2)$. 
If $T_2$ has size at most $K_2$, then the constructed mapping $g(S, \cdot)$ is indeed a $(2\varepsilon)$-inverse in $S$, 
but unfortunately, this assumption may not be true. For example, this happens if $S$ is a large diagonal. 
In the proof below, we use a second partition of $S$, and by evaluating $g_1$ a second time on its parts, 
we can prune~$T$ to the required size.

\begin{proof}
 Let $R_1, \dots, R_{K_2}$ be a partition of $S$ where each set $R_i$ has size at most~$K_1$ 
 and contains at most 1 pair of the form $(x_1, \cdot)$ for all $x_1 \in \mcX_1$. 

 Such a set can be obtained by assigning the elements of $S$ as follows. 
For all $x_1 \in \mcX_1$, select the pairs of the form $(x_1, \cdot)$ in $S$, 
 and add them to different sets $R_i$ of smallest cardinality. 
 Note that since $\# S_{x_1} \le K_2$, this is always possible. In this way, 
 the parts $R_i$ are filled such that in each moment, their cardinalities differ by at most one.

  \medskip
 \noindent
 {\em Construction of $g(S, y_1, y_2)$.} 
 Let $T$ be the set of pairs $\tx = (x_1, x_2)$ in $S$ for which 
 \[
  g_1(R'_\tx, y_1) \; =\;  g_1(S_{x_2}, y_1) \; = \; x_1,
 \]
 where $R'_\tx$ is the projection on the first coordinate of the part~$R_i$ that contains~$(x_1, x_2)$.
 Recall that $T_2$ is the projection of $T$ on the second coordinate.
Let the value of $g$ be the pair $(x_1, x_2)$ in $T$ for which $g_2(T_2, y_2) = x_2$, 
provided the output of $g_2$ is defined. Otherwise, let the value be undefined.

\medskip
 \noindent
We prove that with probability at most $3\varepsilon$:
  \[
    g(S, \ff_1(x_1), \ff_2(x_2)) \;\not=\; (x_1, x_2).
 \]
 Recall that $S_{x_2}$ and $R'_\tx$ have at most $K_1$ elements.
$T_2$ has at most $K_2$ elements, since $T$ contains at most $1$ element from each part $R_i$, 
 and there are $K_2$ such parts. If 
  \begin{align*}
   g_1(R'_\tx, \ff_1(x_1)) \; &=\; x_1 \\
  g_1(S_{x_2}, \ff_1(x_1)) \; &= \; x_1 \\
   g_2(T_2, \ff_2(x_2)) \;&=\; x_2, 
 \end{align*}
 then the required output appears. 
 Indeed, the first 2 conditions imply that $(x_1, x_2) \in T$, 
 the last condition selects~$x_2$, and the second condition implies 
 that for this value of $x_2$, there is a unique pair in $T$ 
of the form $(\cdot, x_2)$. Thus the output must be $(x_1, x_2)$. 

  By invertibility and the union bound, 
the probability that one of these events does not happen  is at most $\varepsilon+\varepsilon+\varepsilon$.
The proposition is proven.
\end{proof}

\section{Degree lower bounds for condensers: Proof of Proposition~\ref{prop:degreeLowerboundCondensers}}\label{s:condenserLowerbounds}

\begin{proposition*}[Restated]
  \propDegreeLowerboundCondensers
\end{proposition*}

\begin{proof}
  We first increase $\varepsilon$ to the minimal value such that $2\varepsilon D$ is a positive integer.
  This can be done without affecting the statement, because
  for $\varepsilon = 1/2$ the value is always integer.

  If $y$ is randomly selected according to a distribution on $\mcY$ with $(1/K)$-excess at most~$\varepsilon$, 
  then for every subset $Y \subseteq \mcY$ of size less than $\varepsilon K$, we have $\Pr [y \in Y] < 2\varepsilon$. 
  Therefore, by contraposition, it suffices to prove the following.

  \begin{claim*}
  If the inequality of the proposition is false, then 
  there exists a subset $Y \subseteq \mcY$ of size less than $\varepsilon K$, 
  and $K$ different $x \in \mcX$ for which
  \begin{equation}\tag{*}\label{eq:bad_x}
  \Pr \left[ f\left(x,U_{[D]}\right) \in Y \right] \;\ge\; 2\varepsilon \,.
  \end{equation}
  \end{claim*}

  \noindent
  We construct $Y$ using the probabilistic method. 
  Let $s = \lceil \varepsilon K\rceil - 1$.
  Select $Y$ randomly among all 
  $s$-element subsets.
  For a fixed $x \in \mcX$, we lower bound the probability that $f(x,i) \in Y$ for all $i \in [2\varepsilon D]$.  Note that the most unfavorable case is when $f(x,1), \ldots, f(x, [2\varepsilon D])$ are all different. Therefore, the probability is at least
   \[
 \ge \;\; {\#\mcY - 2\varepsilon D \choose  s- 2\varepsilon D} /  {\#\mcY \choose s}   \;\; 
  = \;\; \frac{s}{\#\mcY} 
    \cdot \frac{s-1}{\# \mcY - 1} \cdot
    \ldots \cdot \frac{s - 2\varepsilon D + 1}{\#\mcY - 2\varepsilon D + 1}. 
  \]
  The last factor is the smallest and its numerator is at least
  $\varepsilon K - 2\varepsilon D\ge \varepsilon K/2$, because our assumptions imply $K \ge 4D$.
  Thus each factor exceeds~$\varepsilon/(2A)$, where $A = \# \mcY/K$. 
  By a similar reasoning, the probability that $f(x,i) \not\in Y$ for $i = 2\varepsilon D + 1, \dots, D$ 
  is at least $(1-\varepsilon)^{(1-2\varepsilon)D}$, where we used $(\# \mcY - \# Y - D)/\#\mcY \ge 1-\varepsilon/2 - \varepsilon/2$,
  by the assumptions $\# \mcY \ge 2K$ and $\# \mcY \ge 2D/\varepsilon$.
  The probability that precisely $2\varepsilon D$ of the elements $f(x,1), \ldots, f(x,[D])$  are selected in $Y$ is at least 
  \[
  \ge \;\; 
  {D  \choose 2\varepsilon D} \cdot
  \left(\frac{\varepsilon}{2A}\right)^{2\varepsilon D} \cdot
  (1-\varepsilon)^{(1-2\varepsilon) D} 
  \;\;\ge\;\; 
  \left( \tfrac{1}{2\varepsilon}\right)^{2\varepsilon D} \cdot \left(\tfrac{\varepsilon}{2A}\right)^{2\varepsilon D} 
  \cdot (1-\varepsilon)^{\tfrac{1}{\varepsilon} \cdot \varepsilon D}
  \;\; \ge \;\;  \left(\frac{2^{-3}}{A}\right)^{2\varepsilon D}.
  \]
  For a randomly selected $Y$, the expected number of elements $x$ satisfying  \eqref{eq:bad_x} is at least
  $\# \mcX$ times this quantity, by additivity of  expectations. 
  This is at least $K$, because by assumption, the negation of the inequality of Proposition~\ref{prop:degreeLowerboundCondensers} holds, and 
  is equivalent to
  \[
    K < \#\mcX \left(\frac{2^{-3}}{A}\right)^{2\varepsilon D}.
  \]
  Since for a random $Y$, the expected number of $x$ satisfying  \eqref{eq:bad_x} is at least $K$, at least 1 set $Y$ exists
  that satisfies the conditions of the claim.  Hence the proposition  is proven.
\end{proof}

\end{document}